\documentclass[11pt,a4paper]{amsart}
\usepackage[utf8]{inputenc}
\usepackage[a4paper,margin=3.0cm]{geometry}
\usepackage{amsmath,amssymb,amsthm,mathtools,abraces}
\usepackage{stmaryrd}
\usepackage{color,xcolor}
\usepackage[hypertexnames=false,hyperfootnotes=false,colorlinks=true,linkcolor=blue,%
citecolor=purple,filecolor=magenta,urlcolor=cyan,unicode,linktocpage=true,pagebackref=false]{hyperref}
\usepackage{yhmath}
\usepackage{verbatim}
\usepackage{enumitem} 
\usepackage{bbm}
\usepackage{csquotes}

\newcommand\be{\begin{equation}}
\newcommand\ee{\end{equation}}

\newcommand\p{\partial}

\newcommand\normordboson{ {\scriptstyle {\genfrac{}{}{0pt}{}{*}{*}}} }
\newcommand\Tr{{\rm Tr}\,}
\newcommand\diag{{\rm diag}\,}
\newcommand{\<}{\left <}
\renewcommand{\>}{\right >}

\def\lvac{\left <0\right |}
\def\rvac{\left |0\right >}
\DeclareMathOperator{\GL}{GL}
\DeclareMathOperator{\Res}{Res}

\DeclareMathOperator{\Id}{Id}
\DeclareMathOperator{\End}{End}
\DeclareMathOperator{\sspan}{span}
\DeclareMathOperator{\can}{can}

\makeatletter
\@namedef{subjclassname@2020}{%
  \textup{2020} Mathematics Subject Classification}
\makeatother

\newtheorem{theorem}{Theorem}
\newtheorem{lemma}{Lemma}[section]
\newtheorem{proposition}[lemma]{Proposition}
\newtheorem{corollary}[lemma]{Corollary}
\newtheorem{remark}{Remark}[section]

\newtheorem*{theorem*}{Theorem}

\theoremstyle{definition}
\newtheorem{definition}{Definition}

\numberwithin{equation}{section}
\title[Cut-and-join operators in CohFT and topological recursion]{
Cut-and-join operators in cohomological field theory and topological recursion}

\author{Alexander Alexandrov}

\address{Center for Geometry and Physics, Institute for Basic Science (IBS), Pohang 37673, Korea
}

\email{ {\tt alexandrovsash at gmail.com}}


\date{\today}

\begin{document}

\begin{abstract} 
We construct a cubic cut-and-join operator description for the partition function of the Chekhov--Eynard--Orantin topological recursion for a local spectral curve with simple ramification points. In particular, this class contains  partition functions of all semi-simple cohomological field theories. The cut-and-join description leads to an algebraic version of  topological recursion. For the same partition functions we also derive N families of the Virasoro constraints and prove that these constraints, supplemented by a deformed 
dimension constraint, imply the cut-and-join description.
\end{abstract}

\maketitle


\tableofcontents

\newpage 


\def\thefootnote{\arabic{footnote}}


\section{Introduction}
\setcounter{equation}{0}


\subsection{Givental decomposition formula}

Cohomological field theories were introduced by Kontsevich and Manin \cite{KM} to describe universal properties of Gromov--Witten theory. A cohomological field theory (CohFT) is given by a family of tensors $\Omega_{g,n}$. These tensors satisfy the compatibility conditions and are labeled by the number of the marked points $n$ and genus $g$ for all stable combinations 
$2g-2+n>0$. For any CohFT $\Omega$ let us consider its partition function, dependent on the formal variables ${\bf T}=\{T^a_k;1\leq a\leq N, k\geq0\}$,
\be
Z_\Omega=\exp\left(\sum_{g,n} \frac{\hbar^{2g-2+n}}{n!} \sum_{k_j} \int_{\overline{\mathcal{M}}_{g,n}} \Omega_{g,n}(e_{a_1}\otimes \dots \otimes e_{a_n} ) \prod_{j=1}^n \psi_j^{k_j} T_{k_j}^{a_j}\right) \in {\mathbb C}[{\bf T}][\![\hbar]\!].
\ee
We also call it the (total) ancestor potential, because for CohFTs associated with the Gromov--Witten theory on some variety, function $Z_\Omega$ is the generating function of the so-called ancestor Gromov--Witten invariants. The most important example here is the Kontsevich--Witten tau-function associated with the Gromov--Witten theory of a point, see Section \ref{Sec2.1}.

Investigation of the Gromov--Witten invariants and, more generally, cohomological field theories, in particular in the higher genera, is a challenging problem. In his seminal works, Givental \cite{Giv1,Giv0} introduced a group action on the CohFT partition functions. The Givental group action can be lifted to the action on the space of CohFTs themselves \cite{Shadrin,Teleman}. In particular, in \cite{Teleman} Teleman proves that an extension of the Givental group, which we call the Givental--Teleman group, acts transitively on the space of semi-simple CohFTs. CohFTs here are not required to be conformal or have a flat unit. It implies that the partition function of an arbitrary semi-simple CohFT can be described by the Givental decomposition formula,
 \be\label{I1}
Z_\Omega=\widehat{R}\widehat{T}\widehat{\Delta} \cdot \prod_{a=1}^N \tau_{1} (\hbar, {\bf T^a}).
\ee
Here $\widehat{R}$ is a Givental group operator given by a quantization of the twisted loop group element \cite{Giv1}, $\widehat{T}$ is a translation operator, and $\widehat{\Delta}$ rescales the topological expansion parameter $\hbar$. The Kontsevich--Witten tau-function $\tau_{1}$ is a solution of the Korteweg--De Vries (KdV) hierarchy which governs the intersection theory of $\psi$ classes on the moduli spaces of punctured Riemann surfaces. 

The Chekhov--Eynard--Orantin topological recursion \cite{EO,EO1} is a universal procedure, which allows us to construct a family of symmetric differentials on a spectral curve with additional structures on it. For many known examples these differentials encode interesting invariants of enumerative geometry and mathematical physics.
 These differentials can also be combined in a canonical way to create a generating function $Z_S$, dependent on the formal variables ${\bf T}$.
 
As it was shown by Dunin-Barkowski, Orantin, Shadrin, and Spitz \cite{DOSS}, partition functions of the semi-simple cohomological field theories are closely related to the Chekhov--Eynard--Orantin topological recursion. Namely, the ancestor potential of a semi-simple CohFT with flat unit can be identified with the partition function of the Chekhov--Eynard--Orantin topological recursion on a certain local spectral curve, $Z_S=Z_\Omega$. It means, in particular, that the partition function of the associated Chekhov--Eynard--Orantin topological recursion is given by the decomposition formula \eqref{I1}. 

This relation was generalized by Chekhov and Norbury \cite{CN}. This generalization describes Givental's decomposition formula for the partition functions on all, possibly irregular, local spectral curves with simple ramifications, that is,  the curves which near ramification points are similar to the Airy curve $(x=\frac{z^2}{2},y=z)$ or the Bessel curve $(x=\frac{z^2}{2},y=\frac{1}{z})$. This generalization also describes some degenerate CohFTs. To describe this generalization one has to consider a slightly more general Givental decomposition formula.
Namely, for a Givental group operator $\widehat{R}$, which in the most general case is not a quantization of the twisted loop group element, a translation operator $\widehat{T}$, and an $\hbar$-rescaling operator $\widehat{\Delta}$
we consider the generalized ancestor potential (see Definition \ref{Defgap} below)
 \be\label{ZB}
Z:=\widehat{R}\widehat{T}\widehat{\Delta} \cdot \prod_{a=1}^N \tau_{\alpha_a} (\hbar, {\bf T^a}).
\ee
Here $\alpha_a\in\{0,1\}$, and the Br\'ezin--Gross--Witten tau-function $\tau_0$ is also a solution of the KdV hierarchy, which is related to certain intersection numbers on the moduli spaces $\overline{\mathcal{M}}_{g,n}$ \cite{NorbS}.
For a local, possibly irregular, spectral curve with $N$ simple ramification points the partition function $Z_S$ of the Chekhov--Eynard--Orantin topological recursion is given by the Givental decomposition formula \eqref{ZB}. The operators $\widehat{R}$, $\widehat{T}$, $\widehat{\Delta}$ can be described in terms of the topological recursion data, see Section \ref{S43}.
The functions $\tau_{\alpha_a}$ are associated with $N$ ramifications points, and  $\alpha_a=1$ ($\alpha_a=0$) for regular (irregular) ramification points.  

\subsection{Main statement}
 
Partition functions of CohFTs and the Chekhov--Eynard--Orantin topological recursion have many interesting properties \cite{Dub,Maninb,EO,EO1,PpandeCFT,CN}, which make them so attractive. In this paper, we establish a new
general property of all such functions. More precisely, we prove that for an arbitrary Givental operator $\widehat{R}$, a translation operator $\widehat{T}$, and a $\hbar$-rescaling operator $\widehat{\Delta}$ the generalized ancestor potential \eqref{ZB} is described by a cubic cut-and-join operator. In particular, it means that this cut-and-join description is valid for the partition functions of semi-simple CohFTs and Chekhov--Eynard--Orantin topological recursion for local, possibly irregular, spectral curves with simple ramifications. This allows us to describe a generalized ancestor potential by an algebraic version of the topological recursion, which, unlike the Chekhov--Eynard--Orantin topological recursion, is linear and is given in terms of polynomials and differential operators. 

Let
\be
\widehat{J}_k^a:=
\begin{cases}
\displaystyle{(2k-1)!!\frac{\p}{\p T_{k-1}^a}} \,\,\,\,\,\,\,\,\,\,\,\, &\mathrm{for} \quad k> 0,\\[3pt]
\displaystyle{\frac{T_{|k|}^a}{(2|k|-1)!!} }&\mathrm{for} \quad k\leq 0.
\end{cases}
\ee

In Section \ref{S5.3}, for a generalized ancestor potential, given by \eqref{ZB}, we construct a cubic operator
\be\label{WIN}
\widehat{W}=\sum_{\substack{i\leq j\leq k,\\ i+j+k \geq 0}}   A_{a,b,c}^{i,j,k} \widehat{J}^a_i \widehat{J}^b_j \widehat{J}^c_k+ \sum_{j=-1}^\infty B_{a}^j \widehat{J}^a_j,
\ee
where the coefficients $A_{a,b,c}^{i,j,k}$ and $B_{a}^j$ are independent of ${\bf T}$ and $\hbar$. We call it the cut-and-join operator. 

The main result of this paper is the following theorem, which describes all generalized ancestor potentials \eqref{ZB} by cubic cut-and-join operators \eqref{WIN}.
\begin{theorem}\label{TIN} 
A generalized ancestor potential satisfies
\be\label{MTf}
Z=\exp\left(\hbar \widehat{W}\right)\cdot 1.
\ee
\end{theorem}
Formula \eqref{MTf} follows from the properties of the operators $\widehat{\Delta}$, $\widehat{T}$, and $\widehat{R}$ and the cut-and-join description of the Kontsevich--Witten and Br\'ezin--Gross--Witten tau-functions \cite{KSS, KS2}; we prove it in Section \ref{S5}. The result was conjectured in \cite[Conjecture 2.1]{H3}. 
By this theorem a generalized ancestor potential is a solution of the cut-and-join equation
\be\label{CAJIN}
\frac{\p}{\p \hbar}Z=\widehat{W}\cdot Z.
\ee
Formula \eqref{MTf} gives a unique solution of this equation in ${\mathbb C}[\![{\bf T}]\!][\![\hbar]\!]$ with $Z\big|_{\hbar=0}=1$.

The cut-and-join description allows us to reconstruct a generalized ancestor potential by a linear recursion. 
Consider the topological expansion
\be
Z =\sum_{k=0}^\infty Z^{(k)}\hbar^k.
\ee
Coefficients $Z^{(k)}$  are polynomials in ${\bf T}$. From Theorem \ref{TIN} we have the following corollary.
\begin{corollary}\label{CorI}
The coefficients of the topological expansion of a generalized ancestor potential satisfy the algebraic topological recursion
\be
Z^{(k)}=\frac{1}{k} \widehat{W}\cdot Z^{(k-1)}
\ee
with the initial condition $Z^{(0)}=1$.
\end{corollary}
The operator $\widehat{W}$ describes a change of topology. This is why we call it the cut-and-join operator.
Cut-and-join formulas should be very convenient computationally because of the linear structure of the algebraic topological recursion.

There are several families of the partition functions in enumerative geometry for which the cut-and-join description of the form \eqref{MTf} is known, see e.g., \cite{MS,KSS, KS2,H3,AWP}.  For these examples the cut-and-join description is closely related to the Virasoro constraints. Therefore, it is natural to expect the cut-and-join operators to be related to the Virasoro constraints for the general setup of Theorem \ref{TIN}.

In Section \ref{S6.1} with any generalized ancestor potential for $a=1\dots,N$, $m\geq -\alpha_a$ we associate the operators
\be\label{VirIN}
\widehat{L}_m^a=\sum_{\substack{{j<k},\\{j+k\geq m+1}}}L^{a;j,k}_{m;b,c} \widehat{J}_j^b\widehat{J}_k^c
-\frac{1}{2\hbar} \sum_{k=m+\alpha_a}  U^{a,b}_{m,k}\frac{\p}{\p T_k^b}
+\frac{\delta_{m,0}}{16}- \frac{\delta_{m,-1}}{2}p_a,
\ee
where the coefficients $L^{a;j,k}_{m;b,c}$, $U^{a,b}_{k,m}$, and $p_a$ are independent of ${\bf T}$ and $\hbar$.
By construction, the operators $\widehat{L}_m^a$ satisfy the commutation relations of a subalgebra of the direct sum of $N$ copies of the Virasoro algebra
\be
\left[\widehat{L}_k^a,\widehat{L}_m^b\right]=\delta^{ab}(k-m)\widehat{L}_{k+m}^a.
\ee
\begin{theorem}
A generalized ancestor potential satisfies the Virasoro constraints
\be\label{LINi}
\widehat{L}_m^a \cdot Z=0, \quad \quad a=1\dots,N,\quad m\geq -\alpha_a.
\ee
\end{theorem}

These Virasoro constraints \label{LIN} are closely related to the quantum Airy structure by Kontsevich and Soibelman \cite{KonS}. In particular, the quantum Airy structure describes the topological recursion  for the local spectral curves with regular simple ramification points \cite{ABCD}. We show that  the Virasoro constraints \eqref{LINi} and the  quantum Airy structure associated with them describe also the partition functions for the irregular spectral curves.

To arrive at Theorem \ref{TIN} we need to supplement the Virasoro constraints \eqref{LINi} with an auxillary equation
\be
\hbar \frac{\p}{\p \hbar} Z=\widehat{H}\cdot Z,
\ee
where
\be
\widehat{H}= \sum_{\substack{i<j,\\ i+j\geq 1}} H_{a,b}^{i,j} \widehat{J}^a_i \widehat{J}^b_j 
\ee
and the coefficients $H_{a,b}^{i,j}$ do not depend on ${\bf T}$ and $\hbar$. This equation is a deformation of the dimension constraint for the product of the tau-functions $\tau_\alpha$. We show that the cut-and-join equation of the form \eqref{CAJIN} follows from this equation and the Virasoro constraints \eqref{LINi}.

It is known that the Chekhov--Eynard--Orantin topological recursion, cohomological field theories and
quantum Airy structures are closely related and essentially, under certain additional conditions,  are equivalent to each other. In this paper we partially explain how these general structures are related to the cut-and-join operators, algebraic topological recursion,
and Virasoro constraints.

\subsection{Open problems}

There are many open questions related to the observations made in this paper. Let us briefly describe some of them. 

The Kontsevich--Witten and Br\'ezin--Gross--Witten tau-functions have very nice matrix model descriptions. Moreover, there are two different matrix models for the Br\'ezin--Gross--Witten tau-function \cite{MMS}, namely the unitary integral and the generalized Kontsevich model. However, the matrix models for the partition functions of the Chekhov--Eynard--Orantin topological recursion or cohomological field theories are not known except for a few examples. We hope that the cut-and-join operators, constructed in this paper, should help to derive the matrix integrals for the partition functions given by the Givental decomposition formula \eqref{ZB}. For this purpose it can be helpful to reformulate the obtained cut-and-join operators in the Miwa parametrization. Cubic cut-and-join description implies that the generating functions should be given by certain sum over the graphs with at most cubic vertices, which also can be related to matrix models.

In this paper, we describe how to construct the cut-and-join operators using the elements of Givental's decomposition formula. However, to construct the Givental operators themselves for given spectral curve data or a Frobenius manifold,  and then to find the coefficients of the cut-and-join operator explicitly is, in general, a highly non-trivial task. Therefore, it would be desirable to describe the coefficients of the cut-and-join operators directly in terms of the underlying topological recursion data. For the semi-simple conformal Frobenius manifolds it should be possible to describe the cut-and-join operators in terms of the Frobenius structure. The Virasoro constraints for the calibrated partition functions have a simple explicit form \cite{DZ,Giv1}, we expect that the cut-and-join operators and the $N$-families of the Virasoro constraints can be described with the help of the similar elements. Moreover, in this case it should be possible to extend the cut-and-join and the Virasoro description obtained in this paper to the descendant potentials. 

Any cubic operator of the form \eqref{WIN} generates some function $Z$.
From our construction it is not clear which properties of the cubic operators $\widehat{W}$ correspond  to the space of generalized ancestor potentials. Moreover, since such operators are not unique, there should be a natural way, probably related to the symplectic structure for CohFT, to fix the ambiguity and to classify all generalized ancestor potentials by these operators. 
 We also expect that the cut-and-join operators can describe not only the partition functions, but the semi-simple CohFTs themselves. The cut-and-join operators should provide an alternative formulation of Teleman's classification of semi-simple CohFTs.

We do not discuss integrability in this paper. We expect that the cut-and-join description should be helpful for investigation of integrable properties of the corresponding partition functions. Nice integrable properties, in particular the KP type integrability, should be related to the distinguished representation theory interpretations of the cut-and-join operators and existence of the planar global spectral curve. 

It is very challenging to generalize the construction to non-semi-simple CohFTs  corresponding to the topological recursion with higher ramification points. In this case the partition functions should satisfy the higher W-constraints and we expect them to be described by the families of the  higher degree  cut-and-join operators. 

We hope to come back to these problems in the upcoming publications.

\subsection{Conventions}

We use a non-standard normalization of the variables $T_k$, which contains an additional factor of $\hbar$ 
compared to the standard normalization. This change of normalization corresponds to transition from the {\em genus expansion} of the generating functions
\be
\exp\left(\sum_{g,n\geq 0} \hbar^{2g-2}{\mathcal F}_{g,n}\right),
\ee
where $g$ is the genus, $n$ is the number of the marked points, to the {\em topological expansion} 
\be
\exp\left(\sum_{g,n \geq 0} \hbar^{2g-2+n}{\mathcal F}_{g,n}\right),
\ee
because the latter is more convenient for our needs.

By $\widehat{\cdot}$ we denote the differential operators (actually, some formal completions of the space of polynomial differential operators) acting on the space of functions of the variables ${\bf T}=\{T_k^a\}$, where $k\geq 0$ and $a$, if any, runs an appropriate range. For an operator $\widehat{A}$ and a function $F({\bf T})$ we denote by $\widehat{A}\cdot F$ the result of action of the operator on the function. We denote by $a,b,c,d$ the indices, associated with the $N$-dimensional space $V$. We assume the Einstein summation convention when these indices repeat twice. The index $\alpha$ takes values in $\{0,1\}$.


\subsection*{Organization of the paper} 
The main goal of Sections \ref{S2}--\ref{S4} is to recall the basic ingredients of  cohomological field theories, Givental--Telemann formalism, and the Chekhov--Eynard--Orantin topological recursion; they do not contain any new material. In Section \ref{S5} we construct a cubic cut-and-join description of the generalized ancestor potentials. This class includes the partition functions 
of the Chekhov--Eynard--Orantin topological recursion on (possibly irregular) local spectral curves with simple ramification points and semi-simple CohFTs. In Section \ref{S6} we construct $N$ families of the Virasoro constraints and a deformed dimension constraint for generalized ancestor potential. We use these constraints to derive the cut-and-join description and discuss possible ambiguities of this description.

\subsection*{Acknowledgments}
The author is grateful to Sergey Shadrin and Ga\"etan Borot for the inspiring discussions and to the anonymous referee of the paper for the useful remarks.
This work was supported by the Institute for Basic Science (IBS-R003-D1). The author would like to thank IHES, where the part of this project was done, for the kind hospitality.

\section{Cut-and-join operators for KW and BGW tau-functions}\label{S2}

In this section, we remind the reader of some properties of the tau-functions of the KdV hierarchy, playing the main role in our construction. In particular,  we describe the Virasoro constraints and their relation to the cut-and-join operators. 


\subsection{Intersection numbers and KdV tau-function}\label{Sec2.1}
Denote by $\overline {\mathcal M}_{g,n}$ the Deligne--Mumford compactification of the moduli space $ {\mathcal M}_{g,n}$ of all compact Riemann surfaces of genus~$g$ with~$n$ distinct marked points. It is a non-singular complex orbifold of dimension~$3g-3+n$,  which is empty unless the {\em stability condition}
\begin{gather}\label{stability}
2g-2+n>0
\end{gather}
is satisfied. 

In his seminal paper~\cite{Wit91}, Witten initiated new directions in the study of $\overline{\mathcal{M}}_{g,n}$. For each marking index~$i$ consider the cotangent line bundle ${\mathbb{L}}_i \rightarrow \overline{\mathcal{M}}_{g,n}$, whose fiber over a point $[\Sigma,z_1,\ldots,z_n]\in \overline{\mathcal{M}}_{g,n}$ is the complex cotangent space $T_{z_i}^*\Sigma$ of $\Sigma$ at $z_i$. Let $\psi_i\in H^2(\overline{\mathcal{M}}_{g,n},\mathbb{Q})$ denote the first Chern class of ${\mathbb{L}}_i$. We consider the intersection numbers
\begin{gather}\label{eq:products}
\<\tau_{k_1} \tau_{k_2} \cdots \tau_{k_n}\>_g:=\int_{\overline{\mathcal{M}}_{g,n}} \psi_1^{k_1} \psi_2^{k_2} \cdots \psi_n^{k_n} \in {\mathbb Q} .
\end{gather}
The integral on the right-hand side of~\eqref{eq:products} vanishes unless the stability condition~\eqref{stability} is satisfied, all  $k_i$ are non-negative integers, and the dimension constraint 
\be\label{dim1}
3g-3+n=\sum_{i=1}^n k_i
\ee 
holds true. Let $T_i$, $i\geq 0$, be formal variables and let
\be
\tau_{KW}:=\exp\left(\sum_{g,n} \hbar^{2g-2+n}F_{g,n}\right),
\ee
where
\be
F_{g,n}:=\sum_{k_1,\ldots,k_n\ge 0}\<\tau_{k_1}\tau_{k_2}\cdots\tau_{k_n}\>_g\frac{\prod T_{k_i}}{n!}.
\ee
A parameter $\hbar$ here is introduced to trace the Euler characteristic of the punctured curve $\Sigma$, therefore it is associated with the  {\em topological expansion} of the generating function.

Witten's conjecture \cite{Wit91}, proved by Kontsevich \cite{Kon92}, states that the partition function~$\tau_{KW}$ becomes a tau-function of the KdV hierarchy after the change of variables~$T_n=(2n+1)!!t_{2n+1}$. 
\begin{theorem}[\cite{Kon92}]
The generating function $\tau_{KW}$
is a tau-function of the KdV hierarchy in the variables $ t_k$.
\end{theorem}
Below we call it the Kontsevich--Witten (KW) tau-function. Integrability of the KW tau-function immediately follows from Kontsevich's matrix integral representation, for more detail see \cite{H3_2} and references therein.

\subsection{Norbury's classes and KdV tau-function}\label{S2.1}
Another KdV tau-function we consider is the Br\'ezin--Gross--Witten (BGW) tau-function. This tau-function conjecturally governs the intersection theory on the moduli spaces $\overline{\mathcal{M}}_{g,n}$, but this time with the insertions of the fascinating Norbury's $\Theta$-classes. 
Norbury's $\Theta$-classes are the cohomology classes, $\Theta_{g,n}\in H^{4g-4+2n}(\overline{\mathcal{M}}_{g,n})$. We refer the reader to \cite{Norb,NorbS,KN} for a detailed description. Consider the intersection numbers
\be\label{Tinter}
 \<\tau_{k_1} \tau_{k_2} \cdots \tau_{k_n}\>_g^\Theta= \int_{\overline{\mathcal{M}}_{g,n}}\Theta_{g,n} \psi_1^{k_1} \psi_2^{k_2} \cdots \psi_n^{k_n} \in {\mathbb Q} . 
\ee
Again, the integral on the right-hand side vanishes unless the stability condition~\eqref{stability} is satisfied, all $k_i$ are non-negative integers, and the dimension constraint 
\be\label{dim2}
g-1=\sum_{i=1}^n k_i
\ee 
holds. Consider the generating function of the 
intersection numbers of $\Theta$ and $\psi$ classes,
\be
F^\Theta_{g,n}:= \sum_{k_1,\ldots,k_n\ge 0} \<\tau_{k_1} \tau_{k_2} \cdots \tau_{k_n}\>_g^\Theta  \frac{\prod T_{k_i}}{n!}
\ee
and
\be
\tau_\Theta = \exp\left(\sum_{g,n} \hbar^{2g-2+n}F_{g,n}^\Theta \right).
\ee
For this generating function Norbury has suggested a direct analog of Witten's conjecture \cite[Conjecture 1]{NorbS}, recently proven by Chidambaram, Garcia--Failde, and Giacchetto.
\begin{theorem}[\cite{NorbPr}]
The generating function $\tau_{\Theta}$
is the BGW tau-function of the KdV hierarchy in the variables $t_k$
\be
\tau_\Theta=\tau_{BGW}.
\ee
\end{theorem}

Norbury also shows that $\tau_\Theta$ is related to the super Riemann surfaces.

The BGW model was introduced in lattice gauge theory 40 years ago.  KdV integrability of the BGW model in the weak coupling limit follows from the relation to the generalized Kontsevich model and was established by Mironov, Morozov, and Semenoff \cite{MMS}. For more details on BGW tau-function see also \cite{KS2}.

The KW and BGW tau-functions can be described in terms of the Virasoro constraints which completely specify the generating functions and are indirectly equivalent to the integrability. In the next section we describe these constraints and their solution in terms of the cut-and-join operators. 


\subsection{Virasoro constraints and cut-and-join operators}

Let us introduce the operators
\be
\widehat{J}_k:=
\begin{cases}
\displaystyle{(2k-1)!!\frac{\p}{\p T_{k-1}}} \,\,\,\,\,\,\,\,\,\,\,\, &\mathrm{for} \quad k> 0,\\[3pt]
\displaystyle{\frac{T_{|k|}}{(2|k|-1)!!} }&\mathrm{for} \quad k\leq 0.
\end{cases}
\ee
Consider the {\em bosonic current}
\begin{equation}\label{bc}
\begin{split}
\widehat{J}(z)&:=\sum_{k=0}^\infty\left( \frac{T_{k}}{(2k-1)!!}z^{2k}+\frac{(2k+1)!!}{z^{2k+2}}\frac{\p}{\p T_{k}}\right)\\
&=\sum_{k\in{\mathbb Z}} \frac{\widehat{J}_k}{z^{2k}}.
\end{split}
\end{equation}
Let us consider the Virasoro algebra associated with the bosonic current $\widehat{J}(z)$,
\be
\widehat{L}(z)=:\sum_{k\in {\mathbb Z}} \frac{\widehat{L}_k}{z^{2k+2}},
\ee
where
\be
\widehat{L}(z)=\frac{1}{2} \normordboson \widehat{J}(z)^2\normordboson,
\ee
and the normal ordering $\normordboson\dots\normordboson$ puts all  $\frac{\p}{\p T_k}$ to the right of all $T_k$. The Virasoro operators are given by $\widehat{L}_m=\frac{1}{2}\sum_{j+k=m+1} \normordboson \widehat{J}_j\widehat{J}_k\normordboson$ or
\begin{multline}\label{virfull}
\widehat{L}_m=\frac{1}{2} \sum_{i+j=-m-1} \frac{T_i T_j}{(2i-1)!! (2j-1)!!}+ \sum_{k=0}^\infty \frac{(2k+2m+1)!!}{(2k-1)!!} T_k \frac{\p}{\p T_{k+m}}\\
+\frac{1}{2} \sum_{i+j=m-1} (2i+1)!! (2j+1)!!\frac{\p^2}{\p T_i \p T_j}.
\end{multline}
Here and below we assume that $T_k$ and $\frac{\p}{\p T_k}$ vanish for negative $k$. The operators $\widehat{J}_m$ and $\widehat{L}_m$ satisfy the commutation relations of the Heisenberg--Virasoro algebra
\begin{equation}\label{comrel}
\begin{split}
\left[\widehat{J}_k,\widehat{J}_m\right]&=(2k-1)\delta_{k+m,1},\\
\left[\widehat{L}_k,\widehat{J}_m\right]&=-(2m-1)\widehat{J}_{k+m},\\
\left[\widehat{L}_k,\widehat{L}_m\right]&=2(k-m)\widehat{L}_{k+m}+\frac{1}{6}k(2k^2+1)\delta_{k+m,0}.
\end{split}
\end{equation}

Let us denote
\be
\tau_1={\tau}_{KW}, \quad \quad \tau_0={\tau}_{BGW}.
\ee
Below we always assume $\alpha \in \{0,1\}$.
The dimension constraints \eqref{dim1} and \eqref{dim2} can be represented as
\begin{equation}
\begin{split}\label{dimV}
 \hbar \frac{\p}{\p \hbar}  {\tau}_\alpha(\hbar, {\bf T})=\frac{1}{ 2\alpha+1}\widehat{L}_0\cdot {\tau}_\alpha (\hbar, {\bf T}).
\end{split}
\end{equation}

The tau-functions $\tau_\alpha$ satisfy the {\em Virasoro constraints}
\begin{equation}
\begin{split}\label{VirC}
\widehat{L}_k^{\alpha} \cdot {\tau}_{\alpha} &=0, \quad  k \geq -\alpha, 
\end{split}
\end{equation}
for the Virasoro operators
\begin{equation}\label{Vir11}
\begin{split}
\widehat{L}_k^{\alpha}&=\frac{1}{2}\widehat{L}_{k}-\frac{(2k+2\alpha+1)!!}{2\hbar}\frac{\p}{\p T_{k+\alpha}} +\frac{\delta_{k,0}}{16}.
\end{split}
\end{equation}
These constraints were derived in \cite{GN,Fuku,DVV}. Operators $\widehat{L}_k^{\alpha}$ satisfy the commutation relations
\be
\left[\widehat{L}_k^{\alpha},\widehat{L}_m^{\alpha}\right]=(k-m)\widehat{L}_{k+m}^{\alpha}, \quad \quad k,m \geq -\alpha.
\ee
Combining (\ref{dimV}) with (\ref{VirC}) one arrives at
\be\label{cajalp}
\frac{\p}{\p \hbar} \tau_\alpha=\widehat{W}_\alpha \cdot \tau_\alpha, 
\ee
where
\be\label{leadcaj}
\widehat{W}_\alpha:=\frac{1}{2\alpha+1} \sum_{k=0}^\infty \frac{T_k}{(2k-1)!!} \left(\widehat{L}_{k-\alpha}+\frac{\delta_{k,\alpha}}{8}\right).
\ee
We can also rewrite these operators as
\be\label{Wa}
\widehat{W}_\alpha=\frac{1}{2\alpha+1}\Res_z \frac{1}{z^{2\alpha-1}}\left(\widehat{J}(z)_+\widehat{L}(z)+\frac{1}{8z^2}\widehat{J}(z)\right),
\ee
where for any formal series $g(z)=\sum g_k z^k$ we put $g(z)_+=\sum_{k=0}^\infty g_k z^k$ and $\Res_z g(z) = g_{-1}$.
These operators are 
\begin{equation}
\begin{split}\label{CAJ1}
\widehat{W}_0&=
\frac{1}{2}\sum_{k,m=0}^\infty \frac{(2k+1)!!(2m+1)!!}{(2k+2m+1)!!}T_{k+m+1}\frac{\p^2}{\p T_k \p T_m}\\
&+\sum_{k,m=0}^\infty \frac{(2k+2m+1)!!}{(2k-1)!!(2m-1)!!}T_{k}T_{m}\frac{\p}{\p T_{k+m}}
+\frac{T_0}{8},\\
\widehat{W}_1&=\frac{1}{6}\sum_{k,m=0}^\infty\frac{(2k+1)!!(2m+1)!!}{(2k+2m+3)!!}T_{k+m+2}\frac{\p^2}{\p T_k \p T_m}\\
&+\frac{1}{3}\sum_{k,m=0}^\infty \frac{(2k+2m-1)!!}{(2k-1)!!(2m-1)!!}T_{k}T_{m}\frac{\p}{\p T_{k+m-1}}
+\frac{T_0^3}{6}+\frac{T_1}{24}.
\end{split}
\end{equation}

For $\alpha\in \{0,1\}$ equation \eqref{cajalp} has a unique solution with a given initial condition $\tau_{\alpha}({\bf 0})=1$:
\begin{theorem}[\cite{KSS,KS2}]\label{TKW}
\be\label{cajkdv}
\tau_{\alpha}=\exp\left({\hbar \widehat{W}_{\alpha}}\right)\cdot 1.
\ee
\end{theorem}
Let us describe this solution in more detail.
\begin{definition}
We introduce a ${\mathbb Z}$ grading on the space of formal series in the variables ${\bf T}$ and $\hbar$ and differential operators on this space, by assigning 
\be\label{deg}
\deg T_k =2k+1,\quad\quad \deg \frac{\p}{\p T_k} =-2k-1, \quad \quad \deg \hbar =0.
\ee
\end{definition}
Then $\deg  \widehat{J}_k  =1-2k$, $\deg \widehat{L}_k = -2k$, and $\deg \widehat{W}_\alpha=2\alpha+1$.
Let us consider the {\em topological expansion} of the tau-functions
\be
\tau_{\alpha}= \sum_{k=0}^\infty \tau_\alpha^{(k)}\hbar^k. 
\ee
The coefficients $\tau_\alpha^{(k)}$ are homogenous polynomials in ${\bf T}$ of degree $(2\alpha+1)k$.
By Theorem \ref{TKW} they satisfy a linear recursion relation.
\begin{corollary}\label{cor1}
\be\label{rec1}
 \tau_\alpha^{(k)} =\frac{\widehat{W}_\alpha}{k}\cdot  \tau_\alpha^{(k-1)}.
\ee
\end{corollary}

We call it the {\em algebraic topological recursion} to distinguish it from the Chekhov--Eynard--Orantin topological recursion considered in Section \ref{S4}. Operators $\widehat{W}_\alpha$ describe the change of topology, that's why slightly abusing notation we call such operators the {\em cut-and-join operators}. These operators are given by infinite sums \eqref{CAJ1}, however, on each step of the recursion \eqref{rec1} only a finite number of terms contribute. Hence, the recursion is given by the action of the polynomial differential operators on polynomials.

The main goal of this paper is to prove a generalization of Theorem \ref{TKW} and Corollary \ref{cor1} for a huge family of partition functions, related to cohomological field theories and the Chekhov--Eynard--Orantin topological recursion.


\subsection{From translations to Virasoro group elements}
Let us describe a relation between the action of  translation and Virasoro group operators on the tau-functions $\tau_\alpha$. Let
\be\label{virw}
{\mathtt l}_m=-z^{2m+1}\frac{\p}{\p z}.
\ee
For a Virasoro group element of the form
\be\label{Virg}
\widehat{V}_{\bf \chi}=\exp\left({\sum_{k>0} \chi_k \widehat{L}_k }\right),
\ee
where the Virasoro operators $\widehat{L}_k$ are given by \eqref{virfull} and coefficients $\chi_k\in {\mathbb C}$ are independent of ${\bf T}$,
consider the series
\be
\Xi(\widehat{V}_{\bf \chi}):=e^{\sum_{k>0} \chi_k {\mathtt l}_k } \,z \, e^{-\sum_{k>0} \chi_k {\mathtt l}_k } \in z+z{\mathbb C}[\![z^2]\!].
\ee

For a given set of parameters $\chi_k\in {\mathbb C}$, $k>0$ we denote the  series  associated with the operator ${\widehat{V}}_{\bf \chi}$ by $f(z;{\bf \chi})=\Xi(\widehat{V}_{\bf \chi})$.
We also introduce the inverse formal series $h(z;{\bf \chi})  \in z+z{\mathbb C}[\![z^2]\!]$,
\be\label{hfun}
f(h(z;{\bf \chi});{\bf \chi})=z.
\ee

For any series $f(z)\in z+z{\mathbb C}[\![z^2]\!]$ we construct the corresponding element of the Virasoro group 
\be\label{V}
\widehat{V}=\Xi^{-1}(f(z)),
\ee
where the coefficients $\chi_k$ are defined implicitly by
\be\label{fdef}
f(z)=e^{\sum_{k>0} \chi_k {\mathtt l}_k } \,z \, e^{-\sum_{k>0} \chi_k {\mathtt l}_k }.
\ee

Using the Virasoro constraints \eqref{VirC} we can substitute the action of translation operators on the tau-functions $\tau_\alpha$ by the action of the Virasoro group elements. Let the functions $v^\alpha(z;{\bf \chi})$ and $f(z;{\bf \chi})$ be related by 
\be\label{vfunc}
{v}^\alpha(z;{\bf \chi})=\frac{ z^{2\alpha+1}-f(z;{\bf \chi})^{2\alpha+1}}{2\alpha+1}  \in z^{3+2\alpha}{\mathbb C}[\![z^2]\!].
\ee
Consider the coefficients
\be
\Delta T_k^{\alpha}= [z^{2k+1}](2k+1)!! {v}^\alpha(z;{\bf \chi}).
\ee

\begin{lemma}[\cite{H3}]\label{virtos}
We have
\be
\exp\left({\hbar^{-1} \sum_{k>\alpha} \Delta T_k^{\alpha} \frac{\p}{\p T_{k}}}\right) \cdot \tau_\alpha=\exp\left({\sum_{k>0} \chi_{k} \widehat{L}_{k} }\right) \cdot \tau_\alpha.
\ee
\end{lemma}



\section{Cohomological field theories}\label{S3}

The main goal of this section is to recall the basic ingredients of the Givental--Teleman approach to cohomological field theories. For more details, see \cite{KM,Giv1,Giv0,Teleman,PPZ,PpandeCFT}.

\subsection{Cohomological field theory} Cohomological field theories  were introduced by Kontsevich and Manin \cite{KM}. A cohomological field theory (CohFT) is given by a family of cohomology classes on the moduli spaces of stables curves. CohFTs were introduced to axiomatize  universal properties of Gromov--Witten theory, they are closely related to the structure of the Frobenius manifolds \cite{Dub,Maninb}.

On the Deligne--Mumford moduli space  $\overline {\mathcal M}_{g,n}$ consider the gluing map, which glues the last two marked points 
\be\label{qmap}
q : \overline {\mathcal M}_{g-1,n+2} \longrightarrow \overline {\mathcal M}_{g,n}.
\ee
Similarly, identifying the last marked points on the two separate curves, we have a map
\be\label{pmap}
p :  \overline {\mathcal M}_{g_1,n_1+1} \times  \overline {\mathcal M}_{g_2,n_2+1}  \longrightarrow \overline {\mathcal M}_{g,n}
\ee
with $g=g_1+g_2$ and $n=n_1+n_2$. The images of both maps $q$ and $p$ are in the boundary $\p \overline {\mathcal M}_{g,n} =\overline {\mathcal M}_{g,n} \backslash {\mathcal M}_{g,n}$.

Let $\pi$ be the forgetting map for the marked point $n+1$,
\be\label{forg}
\pi: \overline{\mathcal{M}}_{g,n+1} \longrightarrow \overline{\mathcal{M}}_{g,n}. 
\ee

Consider a pair $(V,\eta)$, consisting of an $N$-dimensional complex vector space $V$ equipped with a non-degenerate symmetric 2-form $\eta$. We assume that $N$ is finite. A cohomological field theory consists of a system of tensors
\be
\Omega_{g,n} \in H^*(\overline {\mathcal M}_{g,n}, {\mathbb Q}) \otimes (V^*)^{\otimes n}
\ee
defined for all stable cases with $2g-2+n>0$. For $n=0$, $\Omega_{g,0}\in H^*(\overline {\mathcal M}_{g}, {\mathbb Q})$. The elements $\Omega_{g,n}$ must satisfy the following conditions

\begin{enumerate}[label=\roman*)]

\item\label{r1} Each $\Omega_{g,n}$ is $S_n$-invariant, where the symmetric group $S_n$ acts by permutation of both the $n$ marked points on $\overline {\mathcal M}_{g,n}$ and the $n$ copies of $V^*$.

\item\label{r2} The collection $\Omega$ is compatible with the gluing maps 
\begin{equation}
\begin{split}
q^* \Omega_{g,n}(v_1\otimes\dots\otimes v_n)&=\Omega_{g-1,n+2}(v_1\otimes\dots\otimes v_n\otimes \Delta),\\
p^* \Omega_{g,n}(v_1\otimes\dots\otimes v_n)&=\Omega_{g_1,n_1+1}\otimes \Omega_{g_2,n_2+1} \left(\bigotimes_{j=1}^{n_1} v_j \otimes \Delta\otimes \bigotimes_{j=n_1+1}^n v_j \right)
\end{split}
\end{equation}
for all $v_j \in V$. Here $\Delta \in V^{\otimes 2}$ is the bivector dual to $\eta$ and $q^*, p^*$ are the pull-backs of \eqref{qmap}, \eqref{pmap}.

\end{enumerate}
If the space $V$ also contains a distinguished {\em unit} element ${{\mathbbm 1}} \in V$, it must also satisfy
\begin{enumerate}[label=\roman*),resume]

\item\label{r3} 
\be\label{id}
 \Omega_{0,3}(v_1\otimes v_2\otimes {\mathbbm 1}) =\eta(v_1,v_2)
 \ee
 and
\begin{equation}\label{flat}
\begin{split}
 \Omega_{g,n+1}(v_1\otimes\dots\otimes v_n\otimes {\mathbbm 1} )&= \pi^*  \Omega_{g,n}(v_1\otimes\dots\otimes v_n )
\end{split}
\end{equation}
for all $v_j \in V$. Here $\pi^*$ is the pull-back of the forgetful map \eqref{forg}. Equation
\eqref{id} is essentially a non-degeneracy condition.
\end{enumerate}

\begin{definition}\label{CohFTdef}
A collection of tensors $\Omega_{g,n}$ satisfying properties \ref{r1} and \ref{r2} is a {\em cohomological field theory} (CohFT). If \ref{r3} is also satisfied, then it is a {\em CohFT with flat unit}. A CohFT is {\em degenerate} if \eqref{id} is not satisfied.
\end{definition}

A CohFT defines a commutative and associative {\em quantum product} $\bullet$ on $V$
\be
\eta(v_1\bullet v_2,v_3)=\Omega_{0,3}(v_1\otimes v_2\otimes v_3).
\ee
If a CohFT has a unit, then $ {\mathbbm 1} $ is the identity for the quantum product. 

With any CohFT we associate its {\em partition function}
\be\label{PF}
Z_\Omega=\exp\left(\sum_{g,n} \frac{\hbar^{2g-2+n}}{n!} \sum_{k_j} \int_{\overline{\mathcal{M}}_{g,n}} \Omega_{g,n}(e_{a_1}\otimes \dots \otimes e_{a_n} ) \prod_{j=1}^n \psi_j^{k_j} T_{k_j}^{a_j}\right)
\in {\mathbb C}[{\bf T}][\![\hbar]\!]
\ee
with respect to a basis $\{e_1,\dots,e_N\}$ of $V$. The partition function depends on the choice of the basis, however, the functions associated with different basis are related to each other by linear changes of variables ${\bf T}$.
Here and below the summation is over all stable cases, and contributions with $2g-2+n\leq 0$ vanish by definition. In the framework of Gromov--Witten theory these partition functions are the ancestor potentials. 

\begin{remark}\label{Rmk21}
A fundamental example of CohFT is given by the trivial CohFT
\be
V={\mathbb C},\quad \eta(e_1,e_1)=1,\quad {\mathbbm 1}=e_1, \quad\Omega_{g,n}(e_1\otimes\dots \otimes e_1)=1.
\ee
Its partition function is given by the Kontsevich--Witten tau-function, described in Section \ref{S2}.
\end{remark}

\subsection{Semi-simple CohFTs and normalized canonical basis}
A finite dimensional algebra given by the quantum product is {\em semi-simple} if there exists a basis $e_a^{\can} \in V$, such that
\be
e_a^{\can} \bullet e_b^{\can}=\delta_{ab} e_a^{\can}.
\ee
No summation over $a$ is assumed here.
A CohFT is semi-simple if the algebra $(V,\bullet)$ is semi-simple. Below in this section we consider only the semi-simple cases. 
In the {\em canonical} basis  $e_a^{\can}$ the form $\eta$ is also diagonal
\be
\eta(e_a^{\can},e_b^{\can})=\frac{\delta_{a,b}}{\Delta_a}.
\ee
\begin{remark}
The set $(V,\eta,\bullet)$ describes a semi-simple Frobenius algebra. By certain translation one can extend this algebra, at least locally, to a Frobenius manifold \cite{Maninb}.
The basis  $e_a^{\can} \in V$ is associated with the {\em canonical} coordinates on the Frobenius manifold. 
\end{remark}

The degree $0$ part 
\be
\Omega^0_{g,n}  \in H^0(\overline {\mathcal M}_{g,n}, {\mathbb Q}) \otimes (V^*)^{\otimes n}
\ee
of a CohFT $\Omega$ is a {\em topological quantum field theory} (TQFT). It is uniquely determined by $\Omega_{0,3}=\Omega_{0,3}^0$ and the bilinear form $\eta$. If $\Omega$ is 
a CohFT with a flat unit, then TQFT $\Omega^0_{g,n} $ is also a CohFT with a flat unit. 
The trivial CohFT, described in Remark \ref{Rmk21}, is a TQFT with $\Delta=1$.

We also introduce a {\em normalized canonical} basis $\tilde{e}_a:=\sqrt{\Delta_a}e_a^{\can}$ where we fix some choice of the sign of the square root. This basis is orthonormal,
\be\label{etanc}
\eta(\tilde{e}_a,\tilde{e}_b)=\delta_{a,b}.
\ee
TQFT corresponding to a semi-simple CohFT is completely described by
\be\label{diag}
\Omega^0_{g,n}(\tilde{e}_{a_1}\otimes\dots \otimes \tilde{e}_{a_n})=
\begin{cases}
\displaystyle{\Delta_a^{g-1+\frac{n}{2}}} \quad \quad &\mathrm{if} \,\,\, \tilde{e}_{a_1}=\dots =\tilde{e}_{a_n}= \tilde{e}_{a},\\[2pt]
\displaystyle{0}&\mathrm{otherwise}. 
\end{cases}
\ee
\begin{proposition}
A semi-simple CohFT always contains a unit element ${\mathbbm 1}$ satisfying \eqref{id}. Moreover, this unit element is always flat for the corresponding TQFT.
\end{proposition}
\begin{proof}
By \eqref{diag}
\be
{\mathbbm 1}=\sum_{a=1}^N \frac{\tilde{e}_{a}}{\sqrt{\Delta_a}}
\ee
satisfies \eqref{id}. For the corresponding TQFT $\Omega^0_{g,n}$ this unit is flat, as follows from \eqref{diag} and \eqref{flat}.
\end{proof}
If not specified otherwise, below we will work in the normalized canonical basis and $V$ can be identified with ${\mathbb C}^N$.

\subsection{Givental--Teleman group}

Givental \cite{Giv1} introduced a group action on the CohFT partition functions. Givental's group action can be lifted to the action on the space of CohFTs themselves \cite{Shadrin,Teleman}. This approach uses the stable graphs; we do not consider it here and focus on the action on the partition functions, for more details see \cite{PpandeCFT,PPZ, Norb}. 

The loop group $LGL(V)$ consists of $GL(V)$ valued formal series in $z$. The elements  $U(z)\in LGL(V)$ satisfying the symplectic condition
\be
U(z) U^\star (-z)=\Id
\ee
constitute the {\em twisted loop group} $L^{(2)}GL(V) \subset LGL(V)$. Here $\star$ denotes the adjoint with respect to the metric $\eta$.

 We consider the so-called upper-triangular and lower-triangular subgroups, respectively $L^{(2)}_+GL(V)$ and $L^{(2)}_-GL(V)$. Below we denote the elements of  $L^{(2)}_+GL(V)$ (respectively $L^{(2)}_-GL(V)$) by R (respectively S).

\begin{definition}\label{Rmat}
An $R$-matrix is an element
\be
R(z)=\Id+\sum_{j=1}^\infty R_j z^j \in \Id + z \End (V)[\![z]\!]
\ee
satisfying the symplectic condition
\be\label{symp}
R(z) R^\star (-z)=\Id.
\ee
An $S$-matrix is an element
\be
S(z)=\Id+\sum_{j=1}^\infty S_j z^{-j} \in \Id + z^{-1} \End (V)[\![z^{-1}]\!]
\ee
satisfying the symplectic condition
\be
S(z) S^\star (-z)=\Id.
\ee
\end{definition}
We have $R(z) \in L^{(2)}_+GL(V)$ and $S(z) \in L^{(2)}_-GL(V)$. The group $L^{(2)}_+GL(V)$ acts on the space of semi-simple CohFTs \cite{Teleman}, $\Omega\rightarrow R\Omega$.

Let $T \in z^2 V[\![z]\!]$. For a CohFT $\Omega$ we define a collection of cohomology classes 
\be\label{TOm}
T\Omega_{g,n}(v_1\otimes\dots\otimes v_n)=\sum_{m\geq 0} \frac{1}{m!} \pi_{m,*} \Omega_{g,n+m}(v_1\otimes\dots \otimes v_n \otimes T(\psi_{n+1})\otimes\dots \otimes T(\psi_{n+m})),
\ee
where $\pi_m: \overline{\mathcal{M}}_{g,n+m} \mapsto \overline{\mathcal{M}}_{g,n}$ is the map forgetting the last $m$ marked points. By the dimensional reason this action is well defined, $T\Omega_{g,n} \in H^*(\overline {\mathcal M}_{g,n}, {\mathbb Q}) \otimes (V^*)^{\otimes n}$. For the partition function of the CohFT $\Omega$ the action \eqref{TOm} is given by the translations of the variables ${\bf T}$, see the next section. 

\begin{proposition}[\cite{Teleman}]
The tensors $R\Omega$ and $T\Omega$ form CohFTs on $(V,\eta)$.
\end{proposition}

In general, the action of $R$-matrices and translations $T$ do not respect the property of a CohFT to have a  flat unit. However, if they act consistently, they do not spoil this property.
\begin{proposition}[\cite{Teleman}]
Let  $\Omega$ be a CohFT with a flat unit on $(V,\eta,{\mathbbm 1})$. Then for an R-matrix $R(z)$ the translation
\be\label{Tflat}
T(z)=z({\mathbbm 1}-R^{-1}(z){\mathbbm 1})
\ee
combined with $R$ define a collection of tensors $R T \Omega$, which form a CohFT with a flat unit.
\end{proposition}
This consistent action on the partition functions was originally introduced by Givental \cite{Giv1, Giv0}.

A semi-simple CohFT with a flat unit is uniquely determined \cite{Teleman} by the following two structures:
\begin{itemize}
\item The TQFT $\Omega^0$ of $\Omega$,

\item The $R$-matrix  $R(z)$.
\end{itemize}

In this paper the flat unit does not play any essential role, hence we consider the space of all (partition functions of all) semi-simple CohFTs. This space is described by a more general group.
\begin{definition}
The {\em Givental--Teleman group} is given by a semidirect product $G=z^2V[z] \rtimes L^{(2)}_+GL(V)$.
\end{definition}
The Givental--Teleman group describes an independent action of the Givental operators $R$ and translations $T$.

\begin{theorem}[\cite{Teleman}]\label{TT}
The orbit of the Givental--Teleman group $G$ containing TQFT \eqref{etanc}--\eqref{diag} consists of all CohFTs with the Frobenius algebra $(V,\eta,\bullet)$.
\end{theorem}

\begin{remark}
In practice, it is highly challengeable to find the $R$-matrix explicitly for a given semi-simple CohFT.
The $R$-matrix, associated with a given CohFT by Teleman's Theorem \ref{TT}, is known explicitly only for a limited number of cases,
often with additional structure -- the Euler field \cite{Giv0,Giv1, Dub}. 
\end{remark}

\subsection{Givental--Teleman group action on the partition functions}\label{S34}
Following Givental \cite{Giv1}, on the space $H=V(\!(z^{-1})\!)$ of formal Laurent series in $z$ we consider the symplectic form
\be
\Omega(f,g)=\Res_z \eta(f(-z),g(z)).
\ee

For the natural decomposition 
\be
{ H}={H}_+\oplus{H}_-,
\ee
where the subspaces ${ H}_+=\sspan_{V}\{1,z,z^2,\dots\}$ and ${H}_-=\sspan_{V}\{z^{-1},z^{-2},z^{-3},\dots\}$ are Lagrangian, we introduce the Darboux coordinate system $(p_{k}^a,q_{k,a})$ with
\be
{\mathcal J}(z):=\sum_{k=0}^\infty \left(q_k^a z^{k}\tilde{e}_a  +(-z)^{-k-1}p_{k,a} \tilde{e}^a\right) \in H.
\ee
Here $\tilde{e}^a$ is a basis in $V$ dual to $\tilde{e}_a$ with respect to $\eta$. Since  $\{\tilde{e}_a\}$ is the normalized canonical basis,  we have $\tilde{e}^a = \tilde{e}_a$.
\begin{definition}

An operator $A: H \rightarrow H$ is an {\em infinitesimal symplectic transformation} if 
\be
\Omega(Af,g)+\Omega(f,Ag)=0.
\ee
Operator $B: H \rightarrow H$ is a {\em symplectic transformation} if
\be
\Omega(Bf,Bg)-\Omega(f,g)=0.
\ee
\end{definition} 
After Givental, for an infinitesimal symplectic transformation $A$ on ${H}$ we consider the quadratic Hamiltonian
\be
H_A:=\frac{1}{2}\Omega(A{\mathcal J}, {\mathcal J}).
\ee
This defines a Lie algebra isomorphism:
\be
H_{\left[A,B\right]}=\left\{H_A,H_B\right\},
\ee
where the Poisson bracket is given by
\be
\left\{H_A,H_B\right\}:=\sum_{i=0}^\infty\left(\frac{\p H_A}{\p p_{i,a}}\frac{\p H_B}{\p q_i^a}-\frac{\p H_B}{\p p_{i,a}}\frac{\p H_A}{\p q_i^a}\right).
\ee

Using the standard Weyl quantization we quantize these operators to order $\leq 2$ linear differential operators 
\be\label{Gquant}
\widehat{A}=\widehat{H}_A:=\frac{1}{2}\normordboson\Omega(A \widehat{\mathcal J}, \widehat{\mathcal J})\normordboson,
\ee
where
\be\label{Givcur}
\widehat{\mathcal J}(z)=\sum_{k=0}^{\infty}\left( z^k T_k^a \tilde{e}_a+ (-z)^{-k-1}\frac{\p}{\p T_k^a} \tilde{e}^a\right)
\ee
and $\normordboson\dots \normordboson$ denotes the standard bosonic normal ordering, which puts all $\frac{\p}{\p T_k^a}$ to the right of all $T_k^a$. 
\begin{remark}
Note that we do not introduce the so-called {\em dilaton shift} here. This provides a modification of the original Givental formalism \cite{Giv1, Giv0}, where the dilaton shift plays an essential role.
\end{remark}
This leads a central extension of the original algebra, with the commutator
\be\label{Givcom}
\left[\widehat{A},\widehat{B}\right]=\widehat{\left[A,B\right]}+\mathcal{C}(H_A,H_B),
\ee
where the so-called 2-cocycle term satisfies
\be
\mathcal{C}\left(p_{i,a}p_{j,b},q_{k}^c q_{m}^d\right)=-\mathcal{C}\left(q_k^c q_m^d,p_{i,a}p_{j,b}\right)=\delta_{i,k}\delta_{j,m}\delta_{a,c}\delta_{b,d}+\delta_{i,m}\delta_{j,k} \delta_{a,d}\delta_{b,c}
\ee
and vanishes for all other elements $H_A$. For the symplectic transformations we define
\be
\widehat{e^{A}}:=e^{\widehat{A}}.
\ee

It is easy to see that elements of the twisted loop group $L^{(2)}GL(V)$ describe symplectic transformations. 
For the R-matrices we have
\be
R(z)=\exp(r(z)),
\ee
where $r(z)=\sum_{k=1}^\infty r_k z^k$, $r_k$ are endomorphisms of $V$. The endomorphisms $r_k$ are self-adjoint for $k$ odd and skew-self-adjoint for $k$ even with respect to $\eta$. 
Note, that we are working in the normalized canonical basis, for which $\eta$ is elementary \eqref{etanc}. 
After quantization, using (\ref{Gquant}) we obtain 
\be\label{UTS}
\widehat{R}:=\exp\left(\widehat{r}\right),
\ee
where
\be
\widehat{r}=\sum_{k=1}^\infty  \left( \sum_{j=0}^\infty (r_k)_a^b T_j^a \frac{\p}{\p T_{k+j}^b}+\frac{1}{2}\sum_{j=0}^{k-1} (-1)^{j+1} (r_k)^{ab}\frac{\p^2}{\p T_j^a \p T_{k-j-1}^b}\right).
\ee
Here $(r_k)^{ab} = \eta^{ac}  (r_k)^b_c$. Let us recall that  the Einstein summation convention is assumed  for the indices $a,b,c,d$ associated with the space $V$.

For a translation $T(z)= \sum_{j=2}^\infty \Delta T_j^a z^j \tilde{e}_a  \in z^2 V[\![z]\!]$ we put
\be
\widehat{T}:= \exp\left(\frac{1}{\hbar}  \sum_{j=2}^\infty \Delta T_j^a \frac{\p}{\p T_{j}^a} \right). 
\ee

Let us also introduce the operator $\widehat{\Delta}$ acting on the product of the KW tau-functions as
\be\label{Delt}
\widehat{\Delta} \cdot \prod_{a=1}^N \tau_{1}(\hbar, {\bf{T}}^a) := \prod_{a=1}^N \tau_{1}(\hbar \sqrt{\Delta_a}, {\bf{T}}^a). 
\ee
The right-hand side is a partition function of TQFT described by \eqref{etanc}--\eqref{diag}. 

By  Theorem \ref{TT},  the partition function of a semi-simple CohFT satisfies {\em Givental's decomposition formula}:
\begin{corollary}\label{Cor3.4}
For the CohFT associated with a Frobenius algebra $(V,\eta,\bullet)$ and an element $(R,T)$ of the Givental--Teleman group the partition function is given by
\be\label{PartG}
Z_\Omega=\widehat{R} \widehat{T} \widehat{\Delta} \cdot \prod_{a=1}^N \tau_1(\hbar, {\bf{T}}^a). 
\ee
\end{corollary}

\section{Topological recursion and Givental decomposition for degenerate cases}\label{S4}

In this section, we recall the 
Chekhov--Eynard--Orantin topological recursion and its relation to (degenerate) CohFTs. For more details see \cite{EO,EO1,DOSS,CN,Norb}.

\subsection{Degenerate CohFTs and \texorpdfstring{$\Theta$}--classes}

One can consider various degenerations of CohFTs when some of the properties change in a controllable way. An important class of such degenerate CohFTs can be associated with Norbury's classes. The Norbury classes $\Theta_{g,n}$ were considered in Section \ref{S2.1}. They satisfy the properties \ref{r1} and \ref{r2} of CohFT. Instead of \eqref{flat} we have
\be
\Theta_{g,n+1}=\psi_{n+1}\cdot \pi^* \Theta_{g,n}
\ee
and $\Theta_{0,3}=0$. By Definition \ref{CohFTdef}, the sequence of classes $\Theta$ constitute a degenerate CohFT \cite{NorbS}.

The Norbury classes can be paired with any CohFT. Such a pairing always leads to a degenerate CohFT.
\begin{definition}
For any CohFT $\Omega_{g,n}$ on $(V,\eta)$ define $\Omega_{g,n}^\Theta$ to be the sequence of the tensors $\Omega_{g,n}^\Theta  \in H^*(\overline {\mathcal M}_{g,n}, {\mathbb Q}) \otimes (V^*)^{\otimes n}$ given by
\be
 \Omega_{g,n}^\Theta(v_1\otimes\dots\otimes v_n)=\Theta_{g,n}  \Omega_{g,n}(v_1\otimes\dots\otimes v_n).
\ee
\end{definition}
If $\Omega_{g,n}$ is a CohFT  with flat unit, then
\be
\Omega_{g,n+1}^\Theta(v_1\otimes\dots \otimes{\mathbbm 1})=\psi_{n+1}\cdot \pi^*\Omega_{g,n}^\Theta(v_1\otimes\dots\otimes v_n)
\ee
and $\Omega_{0,3}^\Theta=0$.

The multiplication of a CohFT $\Omega$ by $\Theta$ classes,  $\Omega \mapsto \Omega^\Theta$, commutes with the action of Givental R-matrices and commutes with the action of translations $T$ up to rescaling. Namely, for a CohFT $\Omega$,  for any $R(z)\in L^{(2)}_+GL(V)$ and $T(z)\in z V[\![z]\!]$ we have \cite{Norb}
\be
(R\cdot \Omega)^\Theta= R\cdot   \Omega^\Theta, \quad (zT\cdot \Omega)^\Theta= T\cdot   \Omega^\Theta.
\ee 

For the family of classes $\Omega^\Theta$ we define a partition function $Z_{\Omega^\Theta}({\bf T})$ by \eqref{PF}. 
\begin{proposition}[\cite{Norb}]\label{TheCoh}
For a semi-simple CohFT $\Omega$ with the partition function \eqref{PartG}, the partition function of $\Omega^\Theta$ is given by
\be\label{degpf}
Z_{\Omega^\Theta}=\widehat{R} \widehat{T^\Theta} \widehat{\Delta} \cdot \prod_{a=1}^N \tau_\Theta(\hbar, {\bf{T}}^a), 
\ee
where $T^{\Theta}(z)=z^{-1}T(z)\in zV[\![z]\!]$.
\end{proposition}
Here the operator $\widehat{\Delta}$ acts on the product of the tau-functions by \eqref{Delt}.

Note that we use the BGW tau-functions instead of the KW tau-functions used for the regular case \eqref{PartG}.
The role of the BGW tau-function in the Givental decomposition construction was first indicated in \cite{AMM1,AMM2}. 

Below we will show that the partition functions of the degenerate CohFTs \eqref{degpf}, as well as the more general partition functions associated with the Chekhov--Eynard--Orantin topological recursion, can be described by the cut-and-join formulas. 


\subsection{Chekhov--Eynard--Orantin topological recursion}
CohFTs, including the degenerate ones associated with the Norbury classes, are closely related to another universal construction of mathematical physics. The Chekhov--Eynard--Orantin topological recursion allows us to recover an infinite tower of correlation functions starting from a (not necessarily globally defined) {\em spectral curve} $C$, two meromorphic differentials on it, $\hbox{d} x$ and $\hbox{d} y$,  and a canonical bilinear differential $B(z_1,z_2)$ for $z_1,z_2 \in C$. We call the set
\be
S=\left(C, B, x, y\right)
\ee
the {\em topological recursion data}. For the specific choices of the topological recursion data obtained differentials encode different invariants of mathematics and physics, e.g. ones related to the Gromov--Witten theory, knots, topological strings and Hurwitz numbers.  Let us briefly review the formalism of topological recursion, see \cite{EO,EO1,CN} for more details.

Let $C$ be a Riemann surface (actually, it can be an open subset of a Riemann surface, possibly non-compact, that is, the spectral curve can be {\em local}). On this curve we consider two functions, $x$ and $y$; sometimes the spectral curve can be specified by these functions. More specifically, we consider the differential $\hbox{d} x $  with an assumption that it is a meromorphic differential with finitely many simple critical points $p_1,\dots,p_N$. Below we call these points the {\em simple ramification points}. We assume that the function $y$ is meromorphic in the neighbourhoods of the zeros of $\hbox{d} x $. We allow the spectral curve to be {\em irregular} \cite{CN}, that is, the singularities of $y$ (which  are at most simple poles) may coincide with the zeroes of  $\hbox{d} x$. The bidifferential $B$ is symmetric on $C \times C$, with a double pole on the diagonal with double residue equal to 1, and has no other poles in the neighbourhoods of the zeros of $\hbox{d} x $. 
In particular, on $\mathbb{C}\mathrm{P}^1$ the canonical bilinear differential is the Cauchy differentiation kernel
\be\label{Cauch}
B_C(z_1,z_2)=\frac{\hbox{d} z_1 \hbox{d} z_2}{(z_1-z_2)^2},
\ee
where $z$ is a global coordinate.

With this input, on the Cartesian product $C^{n}$ we construct a system of symmetric differentials (or {\em correlators}) $\omega_{g,n}$, $g\geq 0$, $n\geq 1$,  given by 
\begin{align}
	& \omega_{0,1}(z_1) =  y(z_1)  \hbox{d} x (z_1) ; 
	\\ \notag
	& \omega_{0,2}(z_1,z_2) = B(z_1,z_2);
\end{align}
where $z$ are some local coordinates in the neighbourhoods of the zeros of $\hbox{d} x$.
For $2g-2+n>0$ we use the {\em Chekhov--Eynard--Orantin topological recursion}
\begin{equation}
\begin{split}\label{TRff}
	\omega_{g,n} (z_1,\dots,z_n) \coloneqq \ & \frac 12 \sum_{i=1}^N \Res_{z\to p_i} 
	\frac{\int^z_{\sigma_i(z)} B_C(z_1,\cdot)} {  \omega_{0,1}(z_1) -\omega_{0,1}(\sigma_i(z_1))  }\Bigg(
	\omega_{g-1,n+1}(z,\sigma_i(z),z_{\llbracket n \rrbracket \setminus \{1\}})
	\\
	& + \sum_{\substack{g_1+g_2 = g, I_1\sqcup I_2 = {\llbracket n \rrbracket \setminus \{1\}} \\
			(g_1,|I_1|),(g_2,|I_2|) \not= (0,0) }} \omega_{g_1,1+|I_1|}(z,z_{I_1})\omega_{g_2,1+|I_2|}(\sigma_i(z), z_{I_2})\Bigg),
\end{split}
\end{equation}
where  $\sigma_i$ is the deck transformation of $x$ near $p_i$, $i=1,\dots, N$. For the stable cases $2g-2+n>0$ these differentials are meromorphic. We would not go into the discussion of the version of this topological recursion, for more detail e.g. \cite{EO}. The topological recursion \eqref{TRff} has a natural description in terms of graphs.

The $g$th {\em symplectic invariant} for $g\geq 2$ associated with the spectral recursion data $S$ is defined by
\be
F^{(g)}=\frac{1}{2-2g} \sum_{i=1}^N \Res_{z\to p_i}  \Phi(z) \omega_{g,1}(z),
\ee
where $ \Phi(z)=\int^z  y  \hbox{d} x  $ is a primitive of $y(z)  \hbox{d} x(z)$. Definitions of $F^{(0)}$ and $F^{(1)}$ can be found in \cite{EO}. We will identify $\omega_{g,0}=F^{(g)}$. All correlators $\omega_{g,n}$
 with $ 2- 2g- n < 0$ are called {\em stable}, and the others are called unstable.

For any topological recursion data $S$ for $a=1,\dots,N$ let us define the auxiliary differentials on $C$
\be\label{xia}
\xi^a_0(q)=\Res_{p=p_a} \frac{B(p,q)}{\sqrt{2(x(p)-x(p_a))}}, \quad \xi^a_k = -\hbox{d}  \left(-\frac{\p}{ \p x }\right)^{k-1} \frac{\xi^a_0}{\hbox{d} x},\quad k \in {\mathbb Z}_{>0},
\ee
where $p,q\in C$.
The correlators $\omega_{g,n} (z_1,\dots,z_n)$, obtained by the topological recursion \eqref{TRff}, are polynomials in $\xi^a_k$.
With a topological recursion data $S$ we associate a generating function made of only stable correlators with neglected unstable terms,
 \be
Z_S= \exp\left(\sum_{2g-2+n>0}\frac{\hbar^{2g-2+n}}{n!}\omega_{g,n}(z_1,\dots,z_n)\Big|_{ \xi^a_k=T_k^a}\right).
 \ee

Let us consider the simplest case with one critical point of $\hbox{d} x$, that is, $N=1$. On the genus zero spectral curve $C=\mathbb{C}\mathrm{P}^1$ with a fixed global coordinate $z$, we consider a differential $\hbox{d} x$ with a unique zero. The KW tau-function $\tau_1({\bf t})$ is described by the topological recursion on the Airy curve $x=\frac{1}{2}y^2$ \cite{EO1}
\be
S=\left(\mathbb{C}\mathrm{P}^1,B_C,\frac{1}{2}z^2, z\right).
\ee
It was conjectured in \cite{KS2} and proven in \cite{DN} that the BGW tau-function $\tau_0({\bf t})$ is described by the topological recursion on the Bessel curve $xy^2=\frac{1}{2}$ with the topological recursion data
\be
S=\left(\mathbb{C}\mathrm{P}^1,B_C, \frac{1}{2}z^2, \frac{1}{z}\right).
\ee
For these cases the stable correlators of topological recursion are related to the intersection numbers by
\be
\omega_{g,n}^{\alpha}(z_1,\dots,z_n)=\hbox{d}_1\dots \hbox{d}_n \sum_{a_1,\dots,a_n\ge 0} \int_{\overline{\mathcal{M}}_{g,n}}
\Theta_{g,n}^{1-\alpha}\prod_{j=1}^n \psi_j^{a_j}
\left(-\frac{1}{z_j}\frac{\p}{\p z_j}\right)^{a_j} z_j^{-1}
\ee
for $\alpha=1$ and $\alpha=0$ respectively.

For these two tau-functions translation of variables ${\bf T}$ leads to the change of $y$ described by Proposition 4.8 of Chekhov--Norbury:
\begin{proposition}[\cite{CN}]
The generating function of the topological recursion data
\be
S=\left(\mathbb{C}\mathrm{P}^1,B_C,  \frac{1}{2}z^2, z^{2\alpha-1}+\sum_{k=2\alpha}^\infty y_k z^k\right)
\ee
is obtained via translation of the appropriate tau-function
\be\label{transl}
Z_S=\tau_\alpha\left(\left\{T_k -\frac{(2k-1)!!}{\hbar}y_{2k-1}\right\}\right).
\ee
\end{proposition}
Note that only $T_k$ with $k\geq \alpha+1$ are translated.
\begin{remark}
Actually, the proposition of Chekhov--Norbury is more general and describes $y=\sum_{k=2\alpha-1}^\infty y_k z^k$ with $y_{2\alpha-1}\neq 1$. 
We do not need this more general version here.
\end{remark}


\subsection{From topological recursion to Givental's decomposition}\label{S43}

The general relationship between the Chekhov--Eynard--Orantin topological recursion and intersection numbers on the moduli spaces was investigated by Eynard in  \cite{Eyncohft}. 
In \cite{DOSS} Dunin-Barkowski,  Orantin,  Shadrin, and Spitz prove that the partition function of a semi-simple CohFT with flat unit in the Givental decomposition form coincides with the partition function of the  Chekhov--Eynard--Orantin topological recursion on a certain regular local spectral curve. The identification uses the comparison between two graph expansions. Later Chekhov and Norbury \cite{CN} generalized this identification to the case of CohFTs without flat unit. In general, they correspond to the topological recursion on the irregular spectral curves. In this section we recall the identification of the main ingredients of two theories. We mostly follow \cite{CN,Norb}.

Let $S$ be an arbitrary topological recursion data with local spectral curve (possibly irregular), introduced in the previous section.
With any critical point $p_a$, $a=1,\dots,N$, we associate an index $\alpha_a\in \{0,1\}$ such that $\alpha_a=1$ for regular critical points and $\alpha_a=0$ for the irregular critical points.
In the neighbourhood of a critical point $p_a$ we consider a local coordinate $\zeta$ such that $x=x(p_a)+\zeta^2/2$. Then in this neighbourhood
\be
y=\sum_{k=\alpha_a-1}^\infty y_k^a \zeta^k
\ee
with $y_{2\alpha_a-1}^a\neq0$.
We put 
\be\label{Delttr}
\Delta_a=\left(y_{2\alpha_a-1}^a\right)^{-2}.
\ee

Under the {\em rescaling} $y \mapsto \epsilon y$  the correlators of the topological recursion change as
\be
\omega_{g,n} \mapsto \epsilon^{2-2g-n} \omega_{g,n}.
\ee
Therefore, the generating function $\tau_{\alpha_a} (\hbar \sqrt{\Delta_a}, {\bf T^a})$, obtained by the action of the operator $\widehat{\Delta}$
\be\label{Delta}
\widehat{\Delta} \prod_{a=1}^N \tau_{\alpha_a} (\hbar, {\bf T^a})=\prod_{a=1}^N \tau_{\alpha_a} (\hbar \sqrt{\Delta_a}, {\bf T^a}),
\ee
describes the partition function of the Chekhov--Eynard--Orantin topological recursion on the curve $S=\left(\mathbb{C}\mathrm{P}^1,B_C,\frac{1}{2}z^2, z/\sqrt{\Delta_a}\right)$ if $\alpha_a=1$ or $S=\left(\mathbb{C}\mathrm{P}^1,B_C, \frac{1}{2}z^2, \frac{1}{z\sqrt{\Delta_a}}\right)$
if $\alpha_a=0$.

Let 
\be\label{Ttr}
T^a(z)=
\begin{cases}
\displaystyle{z\left(1 -\sqrt{\frac{\Delta_a}{2\pi z}} \int_{\gamma_a}  \hbox{d} y \,e^{\frac{x(p_a)-x}{z}} \right)},\quad & p_a\,\,\, \mathrm{regular},\\[5pt]
\displaystyle{1 -\sqrt{\frac{\Delta_a}{2\pi z}} \int_{\gamma_a} y \hbox{d} x \,e^{\frac{x(p_a)-x}{z}}},\quad & p_a\,\,\, \mathrm{irregular},
\end{cases}
\ee
where $\gamma_a$ is an arc of the steepest descent contour around the critical point $p_a$ and we take the asymptotic expansion of the integral at small $|z|$. We have $T^a(z)=\sum_{k=1+\alpha_a} \Delta T_k^a z^k$ with 
\be
\Delta T_k^a=-(2k-1)!! \frac{y_{2k-1}^a}{y_{2\alpha_a-1}^a}.
\ee
For the regular points we put $\Delta T_1^a:=0$.
Consider the corresponding translation operator
\be\label{Tgen}
\widehat{T}=\exp\left(\frac{1}{\hbar}\sum_{k=1}^\infty\Delta T_k^a \frac{\p}{\p T_{k}^a}\right).
\ee

Using \eqref{xia} we also introduce the matrix 
\be\label{Rtr}
\left(R^{-1}(z)\right)^a_b=-\sqrt{\frac{z}{2\pi}}\int_{\gamma_b}  \xi^a_0 e^{\frac{x(p_b)-x}{z}}
\ee
and its quantization $\widehat{R}$ given by \eqref{UTS}.
\begin{remark}
To get the CohFT with a flat unit we have to consider the regular spectral curve, where the bidifferential $B$ and the function $y$ are related by
\be
\left(R^{-1}(z)\right)^a_b \frac{1}{\sqrt{\Delta_b}}=\frac{1}{\sqrt{2\pi z}} \int_{\gamma_a}  \hbox{d} y \,e^{\frac{x(p_a)-x}{z}}. 
\ee
In this case $R(z)$ and $T(z)$ satisfy \eqref{Tflat}.
\end{remark}

Givental's decomposition of the partition function for the Chekhov--Eynard--Orantin topological recursion is described by the following theorem.
\begin{theorem}[\cite{DOSS,CN}]\label{DOSST} 
Consider topological recursion data $S=(C,B,x,y)$ with a local (possible irregular) spectral curve $S$. There exist operators $\widehat{R}$, $\widehat{T}$ and $\widehat{\Delta}$ defined by \eqref{Rtr}, \eqref{Tgen}, and \eqref{Delta} determined explicitly by $S$, such that
\be\label{PFS}
Z_S=\widehat{R}\widehat{T}\widehat{\Delta} \cdot \prod_{a=1}^N \tau_{\alpha_a} (\hbar, {\bf T^a}).
\ee
\end{theorem}

Let us stress that for a local spectral curve the matrix \eqref{Rtr} in general does not satisfy the symplectic condition \eqref{symp}, $R\notin L^{(2)}_+GL(V)$. However, the operator $\widehat{R}$ in \eqref{PFS} is defined by \eqref{UTS}.

If for the regular spectral curve the matrix \eqref{Rtr} satisfies the symplectic condition, then the generating function $Z_S$ is identified with a generating function of the associated semi-simple CohFT described by Teleman's Theorem \ref{TT}.
\begin{remark}
It is easy to see that the approach of \cite{DOSS} is applicable to the general case of Theorem \ref{DOSST}, and the inverse statement for the general case is also correct. Namely, for a set of operators $\widehat{R}$, $\widehat{T}$, and $\widehat{\Delta}$ there exists a local (possibly irregular) spectral curve with simple ramification points, such that the partition function of the Chekhov--Eynard--Orantin topological recursion on this curve coincides with the right-hand side of \eqref{PFS}.
\end{remark}

\section{Cut-and-join operators}\label{S5}

In this section, we prove that the partition functions of  semi-simple CohFTs and Chekhov--Eynard--Orantin topological recursion, considered in the previous sections, can be described by cubic cut-and-join operators. It implies, in particular, that these partition functions satisfy the algebraic topological recursion. 

\subsection{Generalized ancestor potentials}\label{S51}

Let us consider $V={\mathbb C}^N$ with an orthonormal basis $\tilde{e}_a$. We choose $\alpha_a\in\{0,1\}$ for $a=1,\dots,N$. For any $R(z)=\exp(r(z))\in \Id + z \End (V)[\![z]\!]$, $T(z)=  \sum_{j=1}^\infty \Delta T_j^a z^j \tilde{e}_a \in zV[\![z]\!]$ with $ \Delta T_1^a=0$ if $\alpha_a=1$, and $\Delta=(\Delta_1,\dots,\Delta_N) \in {\mathbb C}^N$ we consider a Givental operator
\be
\widehat{R}:=\exp\left(
\sum_{k=1}^\infty  \left( \sum_{j=0}^\infty (r_k)_a^b T_j^a \frac{\p}{\p T_{k+j}^b}+\frac{1}{2}\sum_{j=0}^{k-1} (-1)^{j+1} (r_k)^{ab}\frac{\p^2}{\p T_j^a \p T_{k-j-1}^b}\right)\right),
\ee
a translation operator
\be
\widehat{T}:=\exp\left(\frac{1}{\hbar}\sum_{k=1}^\infty\Delta T_k^a \frac{\p}{\p T_{k}^a}\right),
\ee
 and a rescaling operator
\be
\widehat{\Delta} \cdot \prod_{a=1}^N \tau_{\alpha_a}(\hbar, {\bf{T}}^a) := \prod_{a=1}^N \tau_{\alpha_a}(\hbar \sqrt{\Delta_a}, {\bf{T}}^a)
\ee
introduced in \eqref{UTS}, \eqref{Tgen}, \eqref{Delt} respectively. 
\begin{definition}\label{Defgap}
With the operators $\widehat{R}$, $\widehat{T}$, and $\widehat{\Delta}$  we define the {\em generalized ancestor potential}
\be\label{ZSi}
Z({\bf T}):=\widehat{R}\widehat{T}\widehat{\Delta} \cdot \prod_{a=1}^N \tau_{\alpha_a} (\hbar, {\bf T^a}).
\ee
\end{definition}
Below we show that this is a well defined element of ${\mathbb C}[{\bf T}][\![\hbar]\!]$.

If $\alpha_a=1$ for all $a$ and  $R(z)\in L^{(2)}_+GL(V)$, then $(R(z),T(z))$ is nothing but an arbitrary element of the Givental--Teleman group, and $Z$ coincides with the partition function $Z_\Omega$ of the corresponding  semi-simple CohFT in the coordinates associated with normalized canonical basis, see Corollary \ref{Cor3.4}. If $\alpha_a=0$ for all $a$ and  $R(z)\in L^{(2)}_+GL(V)$, then $Z$ coincides with the partition function $Z_{\Omega^\Theta}$ of semi-simple CohFT coupled to $\Theta$ classes, see Proposition \ref{TheCoh}. If $R(z)$, $T(z)$ and $\Delta$ are given by \eqref{Rtr}, \eqref{Ttr},  and \eqref{Delttr} respectively, then $Z$ coincides with the partition function $Z_S$ of the Chekhov--Eynard--Orantin topological recursion on a possibly irregular local spectral curve, see Theorem \ref{DOSST}. 

In this section, we construct a cut-and-join description for the generalized ancestor potential $Z$. The cut-and-join operators for the tau-functions $ \tau_{\alpha_a}$ are described by Theorem \ref{TKW},
their deformations given by the subsequent action of the operators $\widehat{\Delta}$, $\widehat{T}$, and $\widehat{R}$ are described correspondingly in Sections \ref{S5.1}, \ref{S5.2}, and \ref{S5.3}.

\subsection{Cut-and-join operators for semi-simple TQFTs}\label{S5.1}

Let  
\be
\widehat{J}(z)=\diag \left(\widehat{J}^1(z),\dots,\widehat{J}^N(z)\right)
\ee
 be a diagonal matrix with the entries $\widehat{J}^a(z)=\sum \widehat{J}_k^a z^{-2k}$ given by the bosonic current \eqref{bc} acting on the space of functions of ${\bf T}^a$. Let also 
\be
v(z)=\diag\left(\frac{1}{2\alpha_1+1}\frac{1}{z^{2\alpha_1-1}},\dots,\frac{1}{2\alpha_N+1}\frac{1}{z^{2\alpha_N-1}}\right),
\ee
and $\widehat{L}(z)=\normordboson\widehat{J}(z)^2\normordboson$.
We consider the operator
\be
\widehat{W}_D=\Res_z \Tr v(z) \left(\widehat{J}(z)_+\widehat{L}(z)+\frac{1}{8z^2}\widehat{J}(z)\right),
\ee
where we take the trace in the space $V$.
By Theorem \ref{TKW},
\be
\prod_{a=1}^N \tau_{\alpha_a} (\hbar, {\bf T^a})=\exp\left(\hbar \widehat{W}_D\right)\cdot 1.
\ee

Let $\sqrt{\Delta}=\diag(\sqrt{\Delta_1},\dots,\sqrt{\Delta_N})$. We introduce the operator
\be\label{WD}
\widehat{W}_\Delta=\Res_z \Tr v(z) \sqrt{\Delta} \left(\widehat{J}(z)_+\widehat{L}(z)+\frac{1}{8z^2}\widehat{J}(z)\right).
\ee
From the definition of the operator $\widehat{\Delta}$ we have the following lemma.
\begin{lemma}
\be
\widehat{\Delta} \cdot \prod_{a=1}^N \tau_{\alpha_a}(\hbar, {\bf{T}}^a) =\exp\left( \hbar \widehat{W}_\Delta \right) \cdot 1.
\ee
\end{lemma}

\begin{remark}
Operators $\widehat{W}_0$ and $\widehat{W}_1$ in \eqref{cajalp} are of degree $1$ and $3$ respectively. Therefore, in general, the operator $\widehat{W}_\Delta $ is a sum of components of degree $1$ and $3$. 
This operator is homogeneous if all $\alpha_a$ are the same, in this case the degree is equal to $1$ (if all $\alpha_a=0$) or to $3$ (if all $\alpha_a=1$).
\end{remark}

For a semi-simple CohFT  $\alpha_a=1$ for all $a=1,\dots,N$. Therefore, in the normalized canonical basis the partition function of the TQFT for a semi-simple CohFT is given by 
\be\label{ZOO}
Z_{\Omega^0}=\exp \left(\hbar \sqrt{\Delta_a} \widehat{W}_1^a\right)\cdot 1.
\ee

Now, if we consider any basis $\check{e}_a$ in $V$ related to the normalized canonical basis by
\be\label{banew}
\check{e}_b=\Psi^a_b \tilde{e}_a, \quad \Psi \in \GL(V),
\ee
then the formal variables for this basis are related to the formal variables in \eqref{ZOO} by the linear transformation
\be
T^a_k=\Psi^a_b \check{T}^b_k.
\ee
Following Givental \cite{Giv1}, we introduce the operator $\widehat{\Psi}$ responsible for the change of the variables, that acts as follows
\be\label{Psiac}
\left.\widehat{\Psi} \cdot f({\bf T})=f({\bf T})\right|_{T^a_k=\Psi^a_b \check{T}^b_k}.
\ee
Then the partition function of a semi-simple TQFT in the basis \eqref{banew} is equal to
\be
Z_{\Omega^0}({\bf \check{T}})=\widehat{\Psi}\cdot Z_{\Omega^0}({\bf T}).
\ee
Let
\begin{equation}
\begin{split}
\widehat{W}_{\Omega^0}&=\frac{1}{3!}\sum_{a=1}^N\sqrt{\Delta_a}\left(
\Psi_b^a(\Psi^{-1})^c_a(\Psi^{-1})^d_a\sum_{k,m=0}^\infty\frac{(2k+1)!!(2m+1)!!}{(2k+2m+3)!!}\check{T}_{k+m+2}^b\frac{\p^2}{\p \check{T}_k^c \p \check{T}_m^d}
\right.\\
&+
2\Psi_b^a\Psi_c^a(\Psi^{-1})^d_a\sum_{k,m=0}^\infty \frac{(2k+2m-1)!!}{(2k-1)!!(2m-1)!!}\check{T}_{k}^b\check{T}_{m}^c\frac{\p}{\p \check{T}_{k+m-1}^d}
\\
&\left.+\Psi_b^a\Psi_c^a\Psi_d^a \check{T}_0^b\check{T}_0^c\check{T}_0^d+\Psi_b^a\frac{\check{T}_1^b}{4}\right).
\end{split}
\end{equation}
\begin{proposition}\label{PropTQFT}
The partition function of a semi-simple TQFT is given by the cut-and-join operator $\widehat{W}_{\Omega^0}$,
\be
Z_{\Omega^0}({\bf \check{T}})=\exp \left(\hbar \widehat{W}_{\Omega^0}\right)\cdot 1.
\ee
\end{proposition}
\begin{remark}
According to Dubrovin \cite{Dub},
\be
\Omega_{0,3}(\check{e}_b,\check{e}_c,\check{e}_d)= \sum_{a=1}^N\sqrt{\Delta_a} \Psi_b^a\Psi_c^a\Psi_d^a,
\ee
which coincides with a contribution in $Z_{\Omega^0}({\bf \check{T}})$ for $g=0,n=3$ in Proposition \ref{PropTQFT}.
\end{remark}

\subsection{Translations and Virasoro group action}\label{S5.2}
In this section, we consider the operator $\widehat{T}$ in \eqref{ZSi}. Let us construct an element of the Virasoro group, associated with the translation operator $\widehat{T}$ by Lemma \ref{virtos}.
For $a\in \{1,\dots,N\}$ consider
\be
v^a(z)=\sqrt{\Delta_a}\sum_{k=1+\alpha_{a}}^\infty \frac{\Delta T_k^a}{(2k+1)!!}z^{2k+1}\in z^{3+2\alpha_a}{\mathbb C}[\![z^2]\!].
\ee
With these functions we associate 
\be
f^a(z)=z\left(1-(2\alpha_a+1)\frac{v^a(z)}{z^{2\alpha_a+1}}\right)^\frac{1}{2\alpha_a+1}\in z{\mathbb C}[\![z^2]\!].
\ee
By \eqref{V} we construct the elements of the Virasoro group
\be
\widehat{V}^a=\Xi^{-1}(f^a(z)).
\ee
Let 
\be\label{Virm}
\widehat{V}:=\prod_{a=1}^N\widehat{V}^a.
\ee
 Then by Lemma \ref{virtos}, one can substitute a translation operator by an element of the Virasoro group.
\begin{lemma}
\be
 \widehat{T} \widehat{\Delta} \cdot \prod_{a=1}^N \tau_{\alpha_a}(\hbar, {\bf{T}}^a)=\widehat{V}\widehat{\Delta} \cdot \prod_{a=1}^N \tau_{\alpha_a}(\hbar, {\bf{T}}^a).
\ee
\end{lemma}
\begin{proof}
The operators $\widehat{V}$ and $\widehat{\Delta}$ commute, $\left[\widehat{V},\widehat{\Delta}\right]=0$, because $\widehat{V}$ is independent of $\hbar$. For the operator $\widehat{T}$ we have
\be
 \widehat{T} \widehat{\Delta}= \widehat{\Delta} \exp\left(\frac{1}{\hbar}\sum_{a=1}^N\sqrt{\Delta_a}\sum_{k=1+\alpha_a}^\infty\Delta T_k^a \frac{\p}{\p T_{k}^a}\right).
\ee
Then, from Lemma \ref{virtos} we have
\be
 \exp\left(\frac{1}{\hbar}\sum_{b=1}^N\sqrt{\Delta_b}\sum_{k=1+\alpha_b}^\infty\Delta T_k^b \frac{\p}{\p T_{k}^b}\right) \prod_{a=1}^N\tau_{\alpha_a}(\hbar ,{\bf T}^a)=\widehat{V}\cdot \prod_{a=1}^N \tau_{\alpha_a}(\hbar ,{\bf T}^a),
\ee
which completes the proof.
\end{proof}
\begin{remark}
For the partition function of the Chekhov--Eynard--Orantin topological recursion, see Section \ref{S43}, we have
\be
v^a(\zeta)=-\frac{\Delta_a}{2}\left( \Phi(\zeta)- \Phi(-\zeta)\right),
\ee
where $ \Phi(\zeta)=\int^z  y  \hbox{d} x  $ is a primitive of $y(\zeta)  \hbox{d} x(\zeta)$. Here we use the local coordinate $\zeta$ given by $x=x(p_a)+\zeta^2/2$.
\end{remark}

From this lemma it follows that for the generalized ancestor potential \eqref{ZSi} we have an equivalent expression
\be\label{Zalt}
Z=\widehat{R}\widehat{V}\widehat{\Delta} \cdot \prod_{a=1}^N \tau_{\alpha_a} (\hbar, {\bf T^a}).
\ee
\begin{remark}
Translation of the  variables ${\bf T}$ given by operators $\widehat{T}$ in $\tau_\alpha$ is equivalent to the insertion of the Miller--Morita--Mumford tautological $\kappa$ classes into the intersection numbers \eqref{eq:products} and \eqref{Tinter}.
Transformation of the cut-and-join operator in this case is described in \cite{AWP},  below we provide some details.  
\end{remark}

For the Virasoro group elements $\widehat{V}^a$  let us consider diagonal matrices $\widehat{{J}}_V(z)$ and $\widehat{{L}}_V(z)$
 with the diagonal elements
\begin{equation}
\begin{split}
\widehat{{J}}_V^a(z)&=\widehat{V}^a \widehat{J}^a(z)\left( \widehat{V}^a\right)^{-1},  \\
\widehat{{L}}_V^a(z)&= \widehat{V}^a \widehat{L}^a(z)\left( \widehat{V}^a\right)^{-1}. 
\end{split}
\end{equation}
Then for the operator
\be\label{WV}
\widehat{W}_V:=\widehat{V}  \widehat{W}_\Delta\widehat{V}^{-1},
\ee
where the operator $\widehat{W}_\Delta$ is given by \eqref{WD}, we have
\be
\widehat{W}_V=\Res_z \Tr v(z) \sqrt{\Delta} \left(\widehat{J}_V(z)_+\widehat{L}_V(z)+\frac{1}{8z^2}\widehat{J}_V(z)\right).
\ee
Since $\widehat{V}\cdot 1=1$, we have the following lemma.
\begin{lemma}
\be
 \widehat{V} \widehat{\Delta} \cdot \prod_{a=1}^N \tau_{\alpha_a}(\hbar, {\bf{T}}^a)=\exp\left( \hbar \widehat{W}_V \right) \cdot 1.
\ee
\end{lemma}
Because of the upper-triangular structure of the operator $\widehat{V}$ its action on the space ${\mathbb C}[{\bf T}][\![\hbar]\!]$ is well defined.

Let us describe the operator $ \widehat{W}_V$. Conjugation in \eqref{WV} does not spoil the diagonal structure of the operator,
\be
\widehat{W}_V=\sum_{a=1}^N \widehat{W}_V^a.
\ee
Using the commutation relations \eqref{comrel} and the property \eqref{fdef} we arrive at
\begin{equation}
\begin{split}\label{JL}
\widehat{J}_V^a(z)&=h_a'(z) \widehat{J}^a(h_a(z)) ,  \\
\widehat{L}_V^a(z)&=h_a'(z)^2 \widehat{L}^a(h_a(z)) +C_a(z), 
\end{split}
\end{equation}
where $h_a(z)$ is a series inverse to $f_a(z)$, $h_a(f_a(z))=1$, and $C_a(z)=\frac{\sqrt{\Delta_a}\Delta T^a_{1+\alpha_a}}{2(2\alpha_a+3)!!}+O(z^2)\in {\mathbb C}[\![z^2]\!]$. 
If $\alpha_a=0$, then the term $C_a(z)$ does not contribute to the residue in $\widehat{W}_V^a$. Only $C_a(0)$ part  for $\alpha_a=1$ contributes to the residue and we have
\begin{multline}
 \widehat{W}_V^a=\frac{\sqrt{\Delta_a}}{2 \alpha_a+1}\Res_z \frac{1}{z^{2\alpha_a-1}}\left(\left(h_a'(z) \widehat{J}^a(h_a(z))\right)_+\left(h_a'(z)^2 \widehat{L}^a(h_a(z)) +\frac{\sqrt{\Delta_a}\Delta T^a_{2}}{30}\right)\right.\\
 \left.+\frac{1}{8z^2}h_a'(z) \widehat{J}^a(h_a(z))\right).
\end{multline}
This operator can be represented as
\be
\widehat{W}_V^a=\sum_{j=-\alpha_a}^\infty \sum_{i=-j-\alpha_a}^\infty A_{V,a}^{i,j} \widehat{J}^a_i \widehat{L}^a_j + \sum_{j=-\alpha_a}^\infty B_{V,a}^j \widehat{J}^a_j,
\ee
where
\begin{equation}
\begin{split}\label{cof}
A_{V,a}^{i,j}&=\frac{\sqrt{\Delta_a}}{2\alpha_a+1}\Res_z\frac{1}{z^{2\alpha_a-1}} \left(\frac{h_a'(z)}{h_a(z)^{2i}}\right)_+\frac{h_a'(z)^2}{h_a(z)^{2j+2}},\\
B_{V,a}^j&=\frac{1}{2\alpha_a+1}\Res_z \left(\frac{\Delta_a}{30}\Delta T_2^a\delta_{\alpha_a,1}+\frac{\sqrt{\Delta_a}}{8z^{2\alpha_a}}\right)\frac{h_a'(z)}{z h_a(z)^{2j}}.
\end{split}
\end{equation}

In terms of the conjugated bosonic current $\widehat{J}_V(z)$ we have
\be
\widehat{W}_V=\Res_z \Tr v(z) \sqrt{\Delta} \left(\widehat{ J}_V(z)_+\normordboson\widehat{ J}_V(z)\widehat{ J}_V(z)\normordboson+I_V(z)\widehat{ J}_V(z)\right),
\ee
where
\be
I_V(z)=\frac{1}{8z^2}\Id+\frac{1}{30}\diag\left(\sqrt{\Delta_1}\Delta T^1_{2}\delta_{\alpha_1,1},\dots,\sqrt{\Delta_N}\Delta T^N_{2}\delta_{\alpha_N,1}\right).
\ee
We can also rewrite the cut-and-join operator as a normal ordered cubic combination of the operators $\widehat{J}_k^a$,
\be\label{Vasn}
\widehat{W}_V^a=\sum_{\substack{i\leq j\leq k,\\ i+j+k \geq 1-\alpha_a}}  A_{V,a}^{i,j,k} \widehat{J}^a_i \widehat{J}^a_j \widehat{J}^a_k+ \sum_{j=-\alpha_a}^\infty \tilde{B}_{V,a}^j \widehat{J}^a_j,
\ee
where the coefficients $A_{V,a}^{i,j,k}$ and $ \tilde{B}_{V,a}^j$ are finite linear combinations of the coefficients \eqref{cof}. These coefficients are polynomials in $\sqrt{\Delta_a}$ and $\sqrt{\Delta_a} \Delta T^a_k$.

For CohFTs the transformation \eqref{WV} does not affect the leading contributions, associated with TQFT, that is
\be
\widehat{W}_V^a=\widehat{W}_D^a+\dots,
\ee
where by $\dots$ we denote the terms of degree less than $\deg \widehat{W}_D^a=2\alpha_a+1$.

\subsection{Givental operator conjugation}\label{S5.3}
In this section, we describe the action of Givental's operator $\widehat{R}$. From the results of the previous section for the partition function \eqref{ZSi} we have
\be
Z=\widehat{R}\exp\left(\hbar \widehat{W}_V\right)\cdot 1,
\ee
where $\widehat{R}$ is a Givental operator of the form
\be\label{Ropp}
\widehat{R}=\exp\left(\sum_{k=1}^\infty  \left( \sum_{j=0}^\infty (r_k)_a^b T_j^a \frac{\p}{\p T_{k+j}^b}+\frac{1}{2}\sum_{j=0}^{k-1} (-1)^{j+1} (r_k)^{ab}\frac{\p^2}{\p T_j^a \p T_{k-j-1}^b}\right)\right).
\ee

Consider also the matrix
\be
\widehat{\mathcal{J}}(z)=\widehat{R} \widehat  { J}_V(z) \widehat{R}^{-1}.
\ee
In general, this matrix is not diagonal anymore.  Let
\be\label{Star}
\widehat{W}:= \widehat{R} \widehat{W}_V \widehat{R}^{-1}.
\ee
Then from \eqref{WV} we have
\be
\widehat{W}=\Res_z \Tr v(z) \sqrt{\Delta} \left(\widehat{\mathcal J}(z)_+\normordboson\widehat{\mathcal J}(z)\widehat{\mathcal J}(z)\normordboson+{\mathcal I}(z)\widehat{\mathcal J}(z)\right),
\ee
where ${\mathcal I}(z)=I_V(z)-\diag((r_1)^{11},\dots,(r_1)^{NN})$.

Again, one can consider the conjugated operators
\be
 \widehat{R} \widehat{J}^a_i  \widehat{R}^{-1}= \widehat{J}^a_i +\sum_{j=i+1}^\infty c_b^j(r)   \widehat{J}^b_j, 
\ee
where $c_b^j(r)$ are some polynomials in the entries of the matrices $r_k$.
The operator $\widehat{W}$ is cubic in the $\widehat{J}^a_k$, and from \eqref{Vasn} we have
\be\label{WOP}
\widehat{W}=\sum_{\substack{i\leq j\leq k,\\ i+j+k \geq 0}}   A_{a,b,c}^{i,j,k} \widehat{J}^a_i \widehat{J}^b_j \widehat{J}^c_k+ \sum_{j=-1}^\infty B_{a}^j \widehat{J}^a_j,
\ee
where $A_{a,b,c}^{i,j,k}$ and $B_{a}^j$ are linear combinations of the coefficients $A_{V,d}^{i',j',k'}$ and $B_{V,d}^{j'}$ in \eqref{Vasn} with the coefficients given by the polynomials in the entries of the matrices $r_{\ell}$. Because of the upper-triangular structure of the operator $\widehat{R}$ these coefficients are well defined. The coefficients $A_{a,b,c}^{i,j,k}$ and ${B}_{a}^j$ are independent of ${\bf T}$ and $\hbar$.

Note that the  components $\widehat{J}^a_j$ in the operator $\widehat{W}$ are normally ordered. For other orderings it is possible
 to get rid of the linear terms in the cut-and-join operator $\widehat{W}$. 

Since $\widehat{R}\cdot 1=1$, we have
\begin{theorem}\label{MTT} 
A generalized ancestor potential satisfies
\be
Z=\exp\left(\hbar \widehat{W}\right)\cdot 1.
\ee
\end{theorem}

\begin{corollary}
The partition function of the Chekhov--Eynard--Orantin topological recursion for a possibly irregular local spectral curve with simple ramification point 
 has a cut-and-join description
 \be\label{ZSS}
Z_S=\exp\left(\hbar  \widehat{W}\right)\cdot 1.
\ee
Here the operator $\widehat{W}$ is of the form \eqref{WOP}, and the coefficients $A_{a,b,c}^{i,j,k}$ and $B_{a}^j$ can be expressed in terms of the topological recursion data $S$.
\end{corollary}
\begin{corollary}
The partition function of a semi-simple CohFT has a cut-and-join description
 \be
Z_\Omega=\exp\left(\hbar  \widehat{W}\right)\cdot 1.
\ee
Here the operator $\widehat{W}$ is of the form \eqref{WOP}, and the coefficients $A_{a,b,c}^{i,j,k}$ and $B_{a}^j$ can be expressed in terms of the underlying TQFT and element of the Givental--Teleman group.
\end{corollary}

\begin{remark}
Note that in this theorem we do not use an exact form of the operator $\widehat{R}$. A cubic cut-and-join description exists if we substitute $\widehat{R}$ with any operator
\be
 \exp\left( \sum_{0\leq k <j} \left((\tilde{r}_{j,k})_a^b T_k^a \frac{\p}{\p T_{j}^b}+ (\tilde{q}_{j,k})^{ab}\frac{\p^2}{\p T_j^a \p T_{k}^b}\right)\right)
\ee
for arbitrary $\tilde{r}_{j,k}$ and $\tilde{q}_{j,k}$.
\end{remark}

\begin{remark}
If all $\alpha_a=0$, then 
\be
\widehat{W}=\sum_{\substack{i\leq j\leq k,\\ i+j+k \geq 1}}   A_{a,b,c}^{i,j,k} \widehat{J}^a_i \widehat{J}^b_j \widehat{J}^c_k+ \sum_{j=0}^\infty {B}_{a}^j \widehat{J}^a_j.
\ee
\end{remark}

\begin{remark}
We construct the cut-and-join operator in the normalized canonical basis, however,  the form of the cut-and-join operator \eqref{WOP} is basis-independent in $V.$ A change of basis in $V$ is described by the operator $\widehat{\Psi}$ given by \eqref{Psiac}. For the CohFTs it may correspond to the transformation from the normalized canonical basis to the flat basis and may be important. 

Action of $\widehat{\Psi}$ does not change the structure of the cut-and-join operators, but affects the values of the coefficients. Namely, it acts on the upper indices of $A$ and $B$ by multiplication with $\Psi$ or $\Psi^{-1}$ depending on the sign of the corresponding index. For instance, for non-positive $i$ and positive $j$ and $k$ we have
\be
A_{i,j,k}^{a,b,c} \mapsto  \Psi^a_d (\Psi^{-1})^b_e (\Psi^{-1})^c_f A_{i,j,k}^{d,e,f}.
\ee
\end{remark}
The generalized ancestor potential $Z$ is a solution of the {\em cut-and-join equation}
\be
\frac{\p}{\p \hbar}Z=\widehat{W}\cdot Z.
\ee
Consider the topological expansion of the partition function
\be
Z =\sum_{m=0}^\infty Z^{(m)} \hbar^m.
\ee
The cut-and-join equation has a unique solution of this form with  the initial condition $Z^{(0)}=1$.
By Theorem \ref{MTT} we have the following corollary.
\begin{corollary}
The coefficients of the topological expansion of a generalized ancestor potential satisfy the algebraic topological recursion
\be\label{Zk}
Z^{(k)}=\frac{1}{k} \widehat{W}\cdot Z^{(k-1)}
\ee
with the initial condition $Z^{(0)}=1$.
\end{corollary}

The operator $\widehat{W}$ is given by a sum of infinitely many terms. However, only a finite number of them contribute on any finite step of the recursion \eqref{Zk}. Namely, by construction the degree of all terms in $\widehat{W}$ is less or equal to $3$ (or $1$ if all $\alpha_a=0$). Therefore, $Z^{(m)} \in{\mathbb C}[{\bf T}]$ is a polynomial with top degree $3k$ (or $k$ if all $\alpha_a=0$). This polynomial, in general, is not homogeneous, and we have
\be
Z^{(m)}=\frac{1}{m}\left(\sum_{\substack{i\leq j\leq k\leq\max(\lfloor 3m/2\rfloor -1,0),\\ i+j+k \geq 0}}   A_{a,b,c}^{i,j,k} \widehat{J}^a_i \widehat{J}^b_j \widehat{J}^c_k+ \sum_{j=-1}^{\max({\lfloor 3m/2\rfloor -1,0)}} B_{a}^j \widehat{J}^a_j\right) \cdot Z^{(m-1)}.
\ee

In particular
\be
Z^{(1)}=A_{a,b,c}^{0,0,0}T_0^a T_0^b T_0^c+B_a^0T_0^a+B_a^1 T_1^a
\ee
consists of contributions with $(g,n)=(0,3)$ and $(g,n)=(1,1)$.

\begin{remark}For CohFTs associated with the Gromov--Witten theory the partition function $Z$ is a generating function of the Gromov--Witten ancestor invariants. In this case, to arrive at the gravitational descendants one has to apply a properly chosen element $S\in L^{(2)}_-GL(V)$ supplemented by certain translations. It describes the so-called {\em calibration} of the Frobenius structure. The meaning of the group $L^{(2)}_-GL(V)$ for general CohFTs is not clear yet.

It is easy to see that such action does not respect the form \eqref{WOP} of the cut-and-join operator. However, we expect that for the homogeneous semi-simple Frobenius manifolds with flat unit and the calibration based on the solution of the Dubrovin's connection there exists a cubic cut-and-join operator, which describes the descendant generating function up to the unstable terms. 
\end{remark}

\subsection{Bosonic Fock space formalism}
Theorem \ref{MTT} can be reformulated in terms of the bosonic Fock space. Let us consider the bosonic operators $J_k^a$, $k\in {\mathbb Z}$, $a\in \{1,2,\dots,N\}$ satisfying the commutation relations
\be
\left[{J}_k^a,{J}_m^b\right]=(2k-1)\delta_{k+m,1}\delta_{a,b}.
\ee
By definition, the Fock space $\mathcal F$ is generated by commuting creation operators $J_k^a$, $k\in {\mathbb Z}_{\leq 0}$. These operators act on the vacuum vector $\rvac$, while the operators $J_k^a$ with  $k\in {\mathbb Z}_{>0}$ annihilate the vacuum, 
\be
J_k^a \rvac =0,\quad k>0.
\ee
For the co-vacuum $\lvac$ we have
\be
\lvac J_k^a=0,\quad k\leq0,
\ee
and this property together with the relation $\lvac 1 \rvac =1$ fixes a map ${\mathcal F} \rightarrow {\mathbb C}$. Then the identity \eqref{ZSS} has an equivalent  form 
\be
Z=\lvac e^{\sum_{k\geq 0} \sum_{a=1}^N T_k^a J_{k+1}^a} e^{\sum_{{i\leq j\leq k, i+j+k \geq 0}}   A_{a,b,c}^{i,j,k} {J}^a_i {J}^b_j {J}^c_k+ \sum_{j=-1}^\infty {B}_{a}^j {J}^a_j}\rvac.
\ee


\section{Virasoro constraints}\label{S6}

In this section, we describe $N$ families of the Virasoro constraints satisfied by the generalized ancestor potential \eqref{ZSi}. In particular, it gives the Virasoro constraints for the partition functions of all semi-simple CohFTs and Chekhov--Eynard--Orantin topological recursion for possibly irregular local spectral curves with simple ramification points. 

\begin{remark}
A family of the Virasoro constraints for the case of homogenous semi-simple CohFT with a flat unit  was constructed by Milanov \cite{Mil}. As it is mentioned in \cite{Mil}, homogeneity is not essential for the construction, hence Milanov's construction works for all semi-simple CohFTs with flat unit. 
Here we generalize this result.  
\end{remark}

We also derive a version of the dimension constraint for a generalized ancestor potential and show how the cut-and-join equation follows from this constraint and the Virasoro constraints. Moreover, we describe possible ambiguities of the cut-and-join operator.

\subsection{Virasoro constraints}\label{S6.1}
The tau-functions $\tau_\alpha$ satisfy the Virasoro constraints \eqref{VirC}. 
Let $\widehat {L}_k^{\alpha_a}$ be the Virasoro operator  \eqref{Vir11} acting on the space of functions of variables ${\bf T}^a$. With the generalized ancestor potential $Z$ given by \eqref{ZSi} we associate the operators
\be
\widehat{L}_k^a:=\widehat{R}\widehat{T}\widehat{\Delta} \widehat {L}_k^{\alpha_a}\widehat{\Delta}^{-1}\widehat{T}^{-1} \widehat{R}^{-1}, \quad \quad a=1,\dots,N,\quad k\geq -\alpha_a.
  \ee
Let us show that the operators $\widehat{L}_k^a$ are well defined.

Operator $\widehat{\Delta}$ rescales the parameter $\hbar$, therefore we have
\be\label{DeltaL}
\widehat{\Delta} \widehat {L}_k^{\alpha_a}\widehat{\Delta}^{-1}=\frac{1}{2}\widehat{L}_{k}({\bf T}^a)-\frac{(2k+2\alpha_a+1)!!}{2\hbar\sqrt{\Delta_a}}\frac{\p}{\p T_{k+\alpha_a}^a} +\frac{\delta_{k,0}}{16}.
\ee
Here $\widehat{L}_{k}({\bf T}^a)$ is the Virasoro operator \eqref{virfull} in the variables ${\bf T}^a$.
After conjugation with the translation operator 
\be
\widehat{T}=\exp\left(\frac{1}{\hbar}\sum_{k=1}^\infty\Delta T_k^a \frac{\p}{\p T_{k}^a}\right)
\ee
we have
\begin{multline}
\widehat{T}\widehat{\Delta} \widehat {L}_k^{\alpha_a}\widehat{\Delta}^{-1}\widehat{T}^{-1} =\frac{1}{2}\widehat{L}_{k}({\bf T}^a)-\frac{(2k+2\alpha_a+1)!!}{2\hbar\sqrt{\Delta_a}}\frac{\p}{\p T_{k+\alpha_a}^a} 
\\
+\frac{1}{2\hbar}\sum_{m=1+\alpha_a}^\infty \frac{(2k+2m+1)!!}{(2m-1)!!}\Delta T_m^a \frac{\p}{\p T_{k+m}^a}+
\frac{\delta_{k,0}}{16}.
\end{multline}
Then we have
\begin{multline}\label{Virnew}
\widehat{L}_k^a=\frac{1}{2}\widehat{L}_{k}({\bf T}^a)+ 
\sum_{\substack{m,\ell,\\ \ell-m> k}} X^{a,c;m}_{b;k,\ell} T^b_m \frac{\p}{\p T^c_\ell}+
\sum_{\substack{m,\ell,\\ \ell+m \geq k}} V^{a,b,c}_{k,m,\ell} \frac{\p^2}{\p T_m^b \p T_\ell^c}
\\
-\frac{1}{2\hbar} \sum_{m=k+\alpha_a}  U^{a,b}_{k,m}\frac{\p}{\p T_m^b}
+\frac{\delta_{k,0}}{16}- \frac{\delta_{k,-1}}{2}p_a.
\end{multline}
Here $X^{a,c;m}_{b;k,\ell}$, $V^{a,b,c}_{k,m,\ell}$, $U^{a,b}_{k,m}$, and $p^a=(r_1)^{aa}$ do not depend on $\hbar$ and ${\bf T}$. The leading coefficient $U^{a,b}_{k,m}$ is not affected by conjugation by $\widehat{R}$ and coincides with the one in \eqref{DeltaL},
\be
U^{a,b}_{k,k+\alpha_a}=\delta_{a,b}\frac{(2k+2\alpha_a+1)!!}{\sqrt{\Delta_a}}.
\ee
From the upper triangular structure of the operator $\widehat{R}$ it is clear that all coefficients  $X^{a,c;m}_{b;k,\ell}$, $V^{a,b,c}_{k,m,\ell}$ and $U^{a,b}_{k,m}$ are well defined. They are polynomials in the entries of the matrices $r_k$ and depend at most linearly on $\Delta_a^{-1/2}$ and $\Delta T_m^a$. 

The operators satisfy the Virasoro commutation relations
\be\label{BigVc}
\left[\widehat{L}_k^a,\widehat{L}_m^b\right]=\delta^{ab}(k-m)\widehat{L}_{k+m}^a.
\ee
Therefore, we have the following theorem.
\begin{theorem}
A generalized ancestor potential satisfies the Virasoro constraints
\be\label{Virf}
\widehat{L}_m^a \cdot Z=0, \quad \quad a=1\dots,N,\quad m\geq -\alpha_a.
\ee
\end{theorem}

\begin{remark}
A linear change of variables by an operator $\widehat{\Psi}$ of the form \eqref{Psiac} does not affect the statement of the theorem,
but changes the coefficients in the operators $\widehat{L}_k^a$ leaving it quadratic in $\widehat{J}^b_j$.
\end{remark}
When  $R(z)\in L^{(2)}_+GL(V)$, the coefficients of the operators $\widehat{L}_k^a$ can be found with the help of Givental's quantization and the commutation relations \eqref{Givcom}.

The Virasoro constraints of Dubrovin--Zhang and Givental \cite{DZ,Giv1} correspond to the diagonal combination of the operators \eqref{Virnew},
\be
\widehat{L}_m=\sum_{a=1}^N \widehat{L}_m^a.
\ee

For the particular choices of the operators $\widehat{R}$, $\widehat{T}$ and $\widehat{\Delta}$, described in Sections \ref{S3} and \ref{S4}, the partition function \eqref{ZSi} describes the Chekhov--Eynard--Orantin topological recursion and semi-simple CohFTs. Therefore, we immediately have the following corollaries.

\begin{corollary} The partition function of the Chekhov--Eynard--Orantin topological recursion for a (possibly irregular) local spectral curve with simple ramification points
satisfies the Virasoro constraints \eqref{Virf},
\be
\widehat{L}_m^a \cdot Z_S=0, \quad \quad a=1\dots,N,\quad m\geq -\alpha_a.
\ee
\end{corollary}

\begin{corollary}
The partition function of a semi-simple CohFT satisfies the Virasoro constraints \eqref{Virf},
\be
\widehat{L}_m^a \cdot Z_\Omega=0, \quad \quad a=1\dots,N,\quad m\geq -\alpha_a.
\ee
\end{corollary}
The last corollary is a generalization of the results by Milanov \cite{Mil} for the cases without flat unit.

\subsection{Uniqueness of the solution}

By construction, a solution for the Virasoro constraints \eqref{Virf} exists. 
Let us show that the Virasoro constraints \eqref{Virf} have a unique (up to a normalization) solution with topological expansion. To prove it we need a simple lemma.

Let for some finite $\tilde{k}<\infty$,
\be
\widehat{O}=\sum_{k=-\infty}^{\tilde{k}} \widehat{O}_k
\ee
be a differential operator in ${\bf T}$  where $\deg \widehat{O}_k =k$ (see \eqref{deg}). 
\begin{lemma}\label{luni}
If the operator $ \widehat{O}$ is independent of $\hbar$,
then the equation
\be
\sum_{k=0}^\infty (2k+1) T_k^a \frac{\p}{\p T_k^a } Z= \hbar\, \widehat{O}\cdot Z
\ee
has a unique, up to normalization, solution 
with the topological expansion, $Z\in {\mathbb C}[\![{\bf T}]\!][\![\hbar]\!]$. Moreover, the coefficients of the topological expansion of the solution are polynomials, $Z\in {\mathbb C}[{\bf T}][\![\hbar]\!]$.
\end{lemma}
\begin{proof}
For a solution with the topological expansion
\be
Z({\bf T})=\sum_{m=0}^\infty Z^{(m)} \hbar^m 
\ee
the coefficients satisfy
\be\label{eqq}
\sum_{k=0}^\infty (2k+1) T_k^a \frac{\p}{\p T_k^a }  Z^{(m)} = \widehat{O} \cdot Z^{(m-1)}.
\ee
For $m=0$ we have
\be
\sum_{k=0}^\infty (2k+1) T_k^a \frac{\p}{\p T_k^a }  Z^{(0)} =0,
\ee
therefore $Z^{(0)}\in {\mathbb C}$ does not depend on ${\bf T}$. Let us show that all higher $Z^{(m)}$ are polynomials in ${\bf T}$, which can be uniquely determined from equation \eqref{eqq} and the initial condition $Z^{(0)}\in {\mathbb C}$. Namely, for $Z^{(m)}=\sum_{k=0}^\infty  Z^{(m)}_k$, where $\deg Z^{(m)}_k=k$, we have
\be
\sum_{k=0}^\infty (2k+1) T_k^a \frac{\p}{\p T_k^a }  Z^{(m)} =\sum_{k=1}^\infty  k Z^{(m)}_k.
\ee
Then the equation 
\be
\sum_{k=0}^\infty k Z^{(m)}_k = \widehat{O} \cdot Z^{(m-1)}
\ee
determines all $Z^{(m)}_k$ from $Z^{(m-1)}$. Since $\tilde{k}<+\infty$, the sum over $k$ is finite and $Z^{(m)}_k=0$ for $k>m\tilde k$, therefore $Z^{(m)}$ is polynomial. The recursion does not define $ Z^{(m)}_0$, therefore $Z({\bf T})$ is defined up to a normalization $Z({\bf 0})\in {\mathbb C}[\![\hbar]\!]$.
\end{proof}
If $\widehat{O}_k=0$ for all $k>0$, then the only solution is $Z({\bf T})=Z({\bf 0})\in{\mathbb C}[\![\hbar]\!]$.

\begin{proposition}\label{Prop6.4}
The Virasoro constraints \eqref{Virf} have a unique, up to normalization $Z({\bf 0})\in {\mathbb C}[\![\hbar]\!]$, solution of the form $Z({\bf T})\in {\mathbb C}[\![{\bf T}]\!][\![\hbar]\!]$.
\end{proposition}
\begin{proof}
For each $a=1,\dots, N$ and $k\geq -\alpha_a$ let us consider a linear combination of the Virasoro operators 
\eqref{Virnew} of the form
\be
\widehat{K}_k^a=\widehat{L}_k^a+\sum_{m=1}^\infty Q^a_{k,m;b}  \widehat{L}_{k+m}^b,
\ee
such that it contains only one term proportional to $\hbar^{-1}$
\begin{multline}
\widehat{K}_k^a=\frac{1}{2}\widehat{L}_{k}({\bf T}^a)+ 
\sum_{\substack{m,\ell,\\ \ell-m> k}} \tilde{X}^{a,c;m}_{b;k,\ell} T^b_m \frac{\p}{\p T^c_\ell}+
\sum_{\substack{m,\ell,\\ m+\ell \geq k}} \tilde{V}^{a,b,c}_{k,m,\ell} \frac{\p^2}{\p T_m^b \p T_\ell^c}
\\
-\frac{1}{2\hbar}\frac{(2k+2\alpha_a+1)!!}{\sqrt{\Delta_a}}\frac{\p}{\p T_{k+\alpha_a}^a}
+\frac{\delta_{k,0}}{16}- \frac{\delta_{k,-1}}{2}\tilde{p}_a.
\end{multline}
The matrices $Q^a_{k,m;b}$ can be obtained from the inversion of the matrices $U^{a,b}_{k,m}$ in \eqref{Virnew} and are well defined because of the triangular structure of the latter. The coefficients $\tilde{X}^{a,c;m}_{b;k,\ell}$, $\tilde{V}^{a,b,c}_{k,m,\ell}$, and $\tilde{p}_a$ are independent of ${\bf T}$ and $\hbar$. 

For these operators we also have
\be
\widehat{K}_k^a \cdot Z=0, \quad \quad a=1,\dots,N,\quad k\geq -\alpha_a.
\ee
Let
\be
\widehat{\tilde{K}}_{k;a}:=\frac{2\sqrt{\Delta_a}}{(2k+1)!!}\widehat{K}_{k-\alpha_a}^a+\frac{1}{\hbar}\frac{\p}{\p T_{k}^a}.
\ee
We have
\be\label{tKe}
\frac{\p}{\p T_{k}^a} Z=\hbar\widehat{\tilde{K}}_{k;a}  \cdot Z, \quad \quad a=1,\dots,N,\quad k\geq 0,
\ee
and
\be
\sum_{k=0}^\infty (2k+1) T_k^a \frac{\p}{\p T_k^a } Z= \hbar\, \widehat{O}\cdot Z,
\ee
where
\be
\widehat{O}=\sum_{k=0}^\infty (2k+1)T_k^a\widehat{\tilde{K}}_{k;a}
\ee
is independent of $\hbar$. This operator $\widehat{O}$ satisfies the conditions of  Lemma \ref{luni}, which concludes the proof.
\end{proof}

\begin{corollary}
The partition function of the Chekhov--Eynard--Orantin topological recursion for a (possibly irregular) local spectral curve with simple ramification points
is determined by the Virasoro constraints \eqref{Virf} up to a normalization.
\end{corollary}

\begin{corollary}
The partition function of a semi-simple CohFT is determined by the Virasoro constraints \eqref{Virf} up to a normalization.
\end{corollary}

\begin{remark}
From Proposition \ref{Prop6.4} it follows that for the generalized ancestor potential the Virasoro constraints \eqref{Virf} imply the existence of the cut-and-join description of the form 
\be
Z=\exp\left(\hbar \widehat{W}\right)\cdot Z_0,
\ee
where $Z_0\in {\mathbb C}[\![\hbar]\!]$ does not depend on ${\bf T}$.
\end{remark}

The commutation relations between the operators $\widehat{K}_k^a$ follow from \eqref{BigVc}. The operators
$\widehat{K}_k^a$ describe the quantum Airy structure of Kontsevich and Soibelman \cite{KonS}. The quantum Airy structure for the Chekhov--Eynard--Orantin topological recursion on the regular spectral curve was described by Andersen, Borot, Chekhov, and Orantin \cite{ABCD}. The Virasoro constraints for the regular spectral curves was earlier considered in \cite{BorVir}. We see that this structure also appears for the irregular spectral curves. Moreover, we see that the quantum Airy structure for the Chekhov--Eynard--Orantin topological recursion on the (possibly irregular) local spectral curves with simple ramifications is equivalent to the Virasoro constraints \eqref{Virf}. 

While the choice of calibration, given by an element of the $L^{(2)}_-GL(V)$, breaks the quantum Airy structure, it does not spoil the Virasoro algebra of linear constraints if the resulting Virasoro operators are well defined.  

We see that the Virasoro constraints, cut-and-join operators, the Chekhov--Eynard--Orantin topological recursion, and CohFTs are closely related to each other. 

\subsection{From Virasoro constraints to cut-and-join operator}
Let
\be
\widehat{H}:=\widehat{R}\widehat{V} \left(\sum_{a=1}^N \frac{\widehat{L}_0({\bf T}^a)}{2\alpha_a+1}\right)\widehat{V}^{-1} \widehat{R}^{-1},
\ee
where $\widehat{V}$ is the Virasoro group element \eqref{Virm}.

\begin{proposition}\label{Lemma6.7}
A generalized ancestor potential satisfies
\be\label{dimd}
\hbar \frac{\p}{\p \hbar} Z=\widehat{H}\cdot Z.
\ee
\end{proposition}
\begin{proof}
By the dimension constraints \eqref{dimV},
\be
 \hbar \frac{\p}{\p \hbar}  \prod_{a=1}^N \tau_{\alpha_a}(\hbar, {\bf{T}}^a) =\sum_{b=1}^N \frac{\widehat{L}_0({\bf T}^b)}{2\alpha_b+1} \cdot \prod_{a=1}^N \tau_{\alpha_a}(\hbar, {\bf{T}}^a).
\ee
The operators $ \hbar \frac{\p}{\p \hbar}$ and $\sum_{a=1}^N \frac{\widehat{L}_0({\bf T}^a)}{1+2\alpha_a}$ commute with the operator $\widehat{\Delta}$, the first of them also commutes with the operators $\widehat{V}$ and $\widehat{R}$. The statement of the proposition follows from \eqref{Zalt}.
\end{proof}
By construction, operator $\widehat{H}$ is bilinear in $\widehat{J}_k^a$,
\be
\widehat{H}= \sum_{\substack{i<j,\\ i+j\geq 1}} H_{a,b}^{i,j} \widehat{J}^a_i \widehat{J}^b_j. 
\ee

Equation \eqref{dimd} fixes the normalization of the partition function $Z({\bf 0})\in {\mathbb C}[\![\hbar]\!]$ up to an $\hbar$-dependent constant $Z({\bf 0})\Big|_{\hbar=0}\in {\mathbb C}$.

\begin{remark}
If $R(z)\in L^{(2)}_+GL(V)$, the operator $\widehat{H}$ corresponds to a certain symplectic transformation and can be obtained by Givental's quantization, considered in Section \ref{S34}.
\end{remark}
With the help of the relations \eqref{tKe} we have
\be
\widehat{H} \cdot Z=\hbar  \sum_{\substack{i<j,\\ i+j\geq 1}}  H_{a,b}^{i,j} (2j-1)!! \widehat{J}^a_i \widehat{\tilde{K}}_{j-1}^b  \cdot Z.
\ee
Therefore, by Proposition \ref{Lemma6.7}
\be\label{newW}
\frac{\p}{\p \hbar} \cdot Z= \sum_{\substack{i<j,\\ i+j\geq 1}}  H_{a,b}^{i,j} (2j-1)!! \widehat{J}^a_i \widehat{\tilde{K}}_{j-1}^b  \cdot Z.
\ee
The operator in the right-hand side is of the form \eqref{WOP}, therefore the cut-and-join description can be derived from the Virasoro constraints \eqref{Virf} and the deformed dimension constraints \eqref{dimd}.

However, it is not clear if the operator in the right-hand side of \eqref{newW} coincides with the operator $\widehat{W}$ in \eqref{WOP}.
Consider the operators 
\be
\widehat{M}_{k,m}^{a,b}=\frac{\p}{\p T_{m}^a}\widehat{\tilde{K}}_k^b-\frac{\p}{\p T_{k}^b}\widehat{\tilde{K}}_m^a.
\ee
We have $\widehat{M}_{k,m}^{a,b}=-\widehat{M}_{m,k}^{b,a}$, therefore, it is enough to consider $k\geq m$.
From \eqref{tKe} we see, that 
\be
\widehat{M}_{k,m}^{a,b}\cdot Z=0.
\ee
For $k\neq m$ the operators $\widehat{M}_{k,m}^{a}$ are non-trivial. They are cubic in $\widehat{J}^a_k$ with coefficients independent of $\hbar$.
Therefore, for arbitrary $C_{a,b}^{k,m}\in {\mathbb C}$  the operator
\be\label{WC}
\widehat{W}(C)=\widehat{W}+\sum_{0\leq m\leq k} C_{a,b}^{k,m}\widehat{M}_{k,m}^{a,b},
\ee
where $\widehat{W}$ is given by \eqref{Star}, describes the cut-and-join operator for the partition function $Z$,
\be
\frac{\p}{\p \hbar} Z=\widehat{W}(C)\cdot Z.
\ee
We expect, that this is the most general form of a cubic cut-and-join operator for a generalized ancestor potential $Z$ and the operator in the right-hand side of \eqref{newW} is given by \eqref{WC} for some $C_{a,b}^{k,m}$.

\bibliographystyle{alphaurl}
\bibliography{KPTRref}

\end{document}